\documentclass{amsart}
\usepackage{amsmath}
\usepackage{amssymb,amsthm}
\usepackage{graphicx}
\newtheorem{thm}{Theorem}[section]

\newtheorem{lem}[thm]{Lemma}

\theoremstyle{definition}
\newtheorem{as}[thm]{Assumption}
\newtheorem{defin}[thm]{Definition}
\newtheorem{rem}[thm]{Remark}

\numberwithin{equation}{section}

\title{Propagation Estimates for Two-cluster Scattering Channels of N-body Schr\"odinger Operators}
\author{Sohei Ashida}
\begin{document}
\maketitle
\begin{abstract}
In this paper we prove propagation estimates for two-cluster scattering channels of N-body Schr\"odinger operators. These estimates are based on the estimate similar to Mourre's commutator estimate and the method of Skibsted \cite{Sk}. We also obtain propagation estimates with better indices using projections onto almost invariant subspaces close to two-cluster scattering channels. As an application of these estimates we obtain the resolvent estimate for two-cluster scattering channels and microlocal propagation estimates in three-body problems without projections. Our method clearly illustrates evolution of the solutions of the Schr\"odinger equation.
\end{abstract}

\section{Introduction}\label{firstsec}
In this paper we prove propagation estimates for two-cluster scattering channels of N-body Sch\"odinger operators denoted by $H$ on C.M-configuration space $X$. The form of these estimates is
$$B(t)\pi_{\alpha}e^{-itH}f(H)\langle x\rangle^{-s'}=O(t^{-s})$$
as $t\to \infty$ in $\mathcal L(L^2(X))$. Here, $f\in C^{\infty}_0(\mathbb R)$, $0\leq s<\min\{s',\rho\}$ with $\rho$ being the index of the decay of the potentials, $\{B(t)\}_{t>0}$ is a family of pseudodifferential operators, $\pi_{\alpha}$ the projection to the eigenspace of two cluster scattering channel $\alpha$, and $X$ the configuration space of the relative position of the particles.

To investigate the large time asymptotics of the solution is one of the main problems of the study of Schr\"odinger equations. The complete classification of the asymptotic behavior called asymptotic completeness was proved for a large variety of potentials (see  e.g. Derezi\'nski-G\'erard \cite{DeGe}). Asymptotic completeness for $N$-body Schr\"odinger equations means that if the initial data is in the subspace of $L^2(X)$ corresponding to the continuous spectrum of $H$, then $N$ particles are decomposed for large times into clusters. Particles in the same cluster form a bound state and different clusters do not interact each other in the limit $t\to \pm\infty$.

To explain asymptotic completeness more accurately we need the notion of cluster decompositions and channels. A cluster decomposition of $N$ particles is a partition of the particles into a family of subsets of the particles which have no intersection with each other (see Section \ref{secondsec}). If $a$ is a cluster decomposition, $X_a$ and $X^a$ denote the subspace of $X$ corresponding to intercluster and internal coordinates respectively. The internal motion in the clusters in the cluster decomposition $a$ is described by the cluster Hamiltonian $H^a$ defined on $L^2(X^a)$. If we label the eigenfunctions of $H^a$ by integers, for a cluster decomposition $a$ and $m$-th eigenfunction, the pair $\alpha=(a,m)$ is called a channel. Asymptotic completeness for short-range potentials suggests that for any $u\in \mathcal H_{ac}$ where $\mathcal H_{ac}$ is a subspace of $L^2(X)$ corresponding to absolutely continuous subspace of $H$,
$$\lim_{t\to \infty}\left\lVert e^{-itH}u-\sum_{\alpha}(e^{it\Delta_a}v_{\alpha})\otimes(e^{-it\lambda_{\alpha}}\psi_{\alpha})\right\rVert=0.$$
Here $\alpha$ takes all the channels with two clusters or more, $\Delta_a$ is the Laplacian on $X_a$, $v_{\alpha}\in L^2(X_a)$ and $\psi_{\alpha}$ is the eigenfunction of $H^a$ corresponding to $\alpha$.

In the proof of asymptotic completeness propagation estimates play important roles. These are the estimates of decay of $\lVert B(t)e^{-itH}u\rVert$ for $u\in L^2(X)$ with respect to the time $t$. The operator $B(t)$ restricts the evolution $e^{-itH}u$ to some forbidden region in the phase space from the viewpoint of classical mechanics (see e.g. G\'erard \cite{Ge} and Skibsted \cite{Sk}).

The propagation estimates we obtain are concerned with the part of the evolution $e^{-itH}u$ asymptotically close to the evolution $(e^{it\Delta_a}v_{\alpha})\otimes(e^{-it\lambda_{\alpha}}\psi_{\alpha})$ with $\alpha=(a,m)$ and $a$ includes only two clusters.  Let $\Pi^a$ and $\Pi_a$ be the orthogonal projection onto $X^a$ and $X_a$ respectively. For any $x\in X$ we define inner coordinates $x^a:=\Pi^a x$ and intercluster coordinates $x_a:=\Pi_a x$. The minimal (resp., maximal velocity estimate Theorem \ref{main1} (resp., Theorem \ref{main4}) means that in the part of the evolution $e^{-itH}u$ as above, $x_a$ goes to infinity with kinetic energy not less (resp., more) than the sum of the total energy and energy obtained from the bound state of clusters. Theorem \ref{main2} means that the position vector and the momentum are parallel. The other estimates Lemma \ref{propagation1} and Lemma \ref{propagation2} have similar meanings. Our method is based on that of Skibsted \cite{Sk} and almost positivity of the commutator of the generator of the dilation in $X_a$ and the Hamiltonian in small energy interval and the range of $\pi_{\alpha}$ which works as Mourre estimate, where $\pi_{\alpha}$ is the projection onto the eigenspace of $H^a$ corresponding to an eigenvalue (see \eqref{myeq3.7.1}).

As an application of the propagation estimates we obtain a resolvent estimate for $\pi_{\alpha}R(\lambda+i0)$ where $R(\lambda+i0)$ is the boundary value of the resolvent of $H$ on the real axis. Resolvent estimates are related to the asymptotic behavior of the resolvent and important because in the many-body problem the scattering matrices which give correspondence of the data in $t\to -\infty$ to that in $t\to +\infty$ of the Schr\"odinger equation are related to the asymptotic behavior of $\pi_{\alpha}R(\lambda+i0)u,\ u\in L^2(X)$ as $\lvert x_a\rvert\to \infty$ (see e.g. Vasy \cite{Va} and also G\^atel-Yafaev \cite{GaYa}).

The resolvent estimate is a microlocalized version of the radiation estimate. In the two-body case the radiation estimates are the estimates that for $-1/2<s<1/2$, $u\in L^{2,s+1}$ and some differential operator $\gamma$ the following holds
$$\lVert \gamma R(\lambda+i0)u\rVert_{L^{2,s}(X)}<C\lVert u\rVert_{L^{2,s+1}(X)},$$
where $L^{2,l}:=\{u\in L^2\vert \langle x\rangle^{l}u\in L^2\}$.
The radiation estimates imply that the leading part of the resolvent as $\lvert x\rvert\to \infty$ is an asymptotically spherical wave which decays as $\lvert x\rvert^{-(n-1)/2}$ where $n$ is the dimension of $X$ (see e.g. Ikebe \cite{Ik}, Isozaki \cite{Is} and Sait\=o \cite{Sa}). In the many-body problem there are several kinds of radiation estimates and the microlocalized estimates in which $\gamma$ is replaced by some pseudodifferential operators with the symbols supported in the region 
$\{(x,\xi)\vert x\cdot \xi<(1-\epsilon)\lvert x\rvert \lvert \xi \rvert\},\ \epsilon>0$ in which the particle can not scatter radially (see e.g. Bommier \cite{Bo}, G\'erard-Isozaki-Skibsted \cite{GIS1,GIS2}, Herbst-Skibsted \cite{HeSk2}, Wang \cite{Wa1,Wa2,Wa3} and Yafaev \cite{Ya1,Ya2}). 

The propagation estimates and the resolvent estimates are related by Fourier-Laplace transform and Stone's formula (see e.g. Herbst-Skibsted \cite{HeSk} and Jensen-Mourre-Perry \cite{JMP}). In \cite{HeSk2} Herbst and Skibsted prove the radiation estimate for the free channel in which each cluster is a single particle. Their method is based on the propagation estimates for the time-dependent Schr\"odinger equation obtained in Skibsted \cite{Sk} and the method in Herbst-Skibsted \cite{HeSk} developed for the two-body problem. As for other channels, in Wang \cite{Wa2} and G\'erard-Isozaki-Skibsted \cite{GIS2} it is proved that for any cluster decomposition $a$ there exists $E>0$ such that for energy $\lambda\geq E$ the microlocalized estimates in the region where the clusters in $a$ are separated hold with the pseudodifferential operator $p(x_a,D_a)$ whose symbol supported in $\{(x_a,\xi_a)\vert x_a\cdot \xi_a<(1-\epsilon)\lvert x_a\rvert \lvert \xi_a \rvert\}$.

Our resolvent estimate for two-cluster channels is included in Wang \cite[Corollary 2.7]{Wa2} and his result is better than ours regarding the indices of the powers of $\langle x\rangle$, dependence of the bound on the energy and the assumption on the potentials. However, the proof of \cite[Corollary 2.7]{Wa2} is incorrect and the indices in the estimate is restricted. Our propagation estimates are better than that is obtained from his resolvent estimates with respect to the indices of time decay. In particular, only our method gives propagation estimates for long-range potentials. Since our proof is based on the propagation estimates, we can see how the estimate is obtained from the viewpoint of the evolution of the solution of time-dependent Schr\"odinger equation, and the propagation estimates themselves have physical meanings.

The content of this paper is as follows. In sec. 2 we introduce some notations and an assumption and state our main results. In sec. 3 and sec. 4 we prove the main results. In the Appendix we collect the results concerned with domains of operators.

\section{Some preliminaries and main results}\label{secondsec}
The N-body Schr\"odinger operator or Hamiltonian is written as
$$-\Delta+\sum_{i<j}V_{ij}(x_i-x_j),$$
on $L^2(X)$ where $x_i\in \mathbb R^{\nu},\ \nu\in \mathbb N\setminus \{0\}$ is the position of the $i$th particle and $X$ is the C.M-configuration space $\left\{x=(x_1,\dots, x_N)\vert x_i\in \mathbb R^{\nu},\sum_{i=1}^Nm_ix_i=0\right\}$ of $N$ particles with masses $m_i$. The inner product in $X$ is given by $x\cdot y=\sum_{i=1}^N2m_ix_i\cdot y_i$. The operator $-\Delta$ denotes the Laplacian in $X$ given by this metric.

A cluster $C$ is a subset of $\{1,2,\dotsm, N\}$, and a family of clusters
$$a=\{ C_1,C_2,\dotsm,C_m\},$$
is called a cluster decomposition if $C_i\cap C_j=\emptyset\ (i\neq j),\ \bigcup_{i=1}^{m}C_i=\{1,2,\dotsm, N\}$. Set for any cluster decomposition $a$
$$X_a=\{x\in X\vert x_i=x_j \ \mathrm{if}\  i,j\in C \ \mathrm{for\ some}\  C\in a\},$$
and denote the orthogonal complement of $X_a$ by $X^a$. $X_a$ (resp., $X^a$) is the intercluster (resp., internal) configuration space. The cluster decomposition
$$(1)\dotsm (\hat i)\dotsm (\hat j)\dotsm (N)(ij),$$
where $\hat\ $ indicates omission is denoted by $(ij)$. We denote by $\Pi^a$ and $\Pi_a$ the orthogonal projections of $X$ onto $X^a$ and $X_a$ respectively. We use the same notations $\Pi^a$ and $\Pi_a$ for the corresponding orthogonal projections of the dual space of $X$. We define for all $x\in X$, $x_a=\Pi_a x$ and $x^a=\Pi^a x$. We denote the differentiation in $\mathbb R^{N\nu}$ projected onto the dual space of $X$ by $\nabla$, and set $\nabla_a=\Pi_a\nabla$ and $\nabla^a=\Pi^a\nabla$. The operators $-\Delta_a$ and $-\Delta^a$ denote the Laplacian in $X_a$ and $X^a$ respectively. If $X_b\subset X_a$, we write $a\leq b$, otherwise $a\nleq b$. We define the cluster Hamiltonian
$$H^a:=-\Delta^a+\sum_{(ij)\leq a}V_{ij}(x_i-x_j).$$

We assume that potentials $V_{ij}(y)$ satisfies the following assumption.
\begin{as}\label{potentialas}
\begin{itemize}
\item[(1)] $V_{ij}(\Delta_y+1)^{-1}$ are compact.
\item[(2)] There exists a constant $R>0$ such that $V_{ij}(y)$ are smooth in the region $\lvert y\rvert >R$ and there exists a constant $\rho>0$ such that
$$\partial^{\alpha}_yV_{ij}(y)=O(\lvert y\rvert^{-\lvert \alpha\rvert-\rho}),\ \lvert y\rvert \to \infty,\ \forall \alpha\in\mathbb N^{\nu}.$$
\end{itemize}
\end{as}

We label the eigenfunctions of $H^a$ by integers $m$, and call pairs $\alpha=(a,m)$ channels. The number of clusters in a cluster decomposition $a$ is denoted by $\# a$.
Given $\epsilon>0$ , we denote by $\chi(x<-\epsilon)$ the smooth function such that $\chi(x<-\epsilon)=1$ for $x<-2\epsilon$ and $\chi(x<-\epsilon)=0$ for $x>-\epsilon$. We denote by $\mathcal H$ the Hilbert space $L^2(X)$.
One of our main theorems is the following propagation estimate.

\begin{thm}[minimal velocity estimate]\label{main1}
Suppose Assumption \ref{potentialas}. Let $a$ be a cluster decomposition such that $\# a=2$, $\lambda_{\alpha}$ be the discrete eigenvalue of $H^a$ corresponding to a channel $\alpha$. Then for any $E>\lambda_{\alpha},\ E\notin \mathcal T:=\bigcup_{\#a\geq 2} \sigma_{\mathrm{pp}}(H^a)$, $\epsilon>0$, any $f\in C_0^{\infty}(\mathbb R)$ supported in a sufficiently small neighborhood of $E$ and any $s,s'\in \mathbb R$ such that $0< s<\min\{s',\rho\}$,
\begin{equation}\label{myeq2.0.1}
\chi\left(\frac{x_a^2}{4t^2}-(E-\lambda_{\alpha})<-\epsilon\right)\pi_{\alpha}e^{-itH}f(H)\langle x\rangle^{-s'}=O(t^{-s}),
\end{equation}
in $\mathcal L(\mathcal H)$ as $t\to \infty$.
\end{thm}
\begin{rem}
In general, the condition $s<\rho$ does not seem to be removable, because  of the interaction between scattering channels (see section \ref{concluding}).
\end{rem}

We introduce the following class of pseudodifferential operators:
Let $S^m_{l},\ m,l \in \mathbb R$, be the symbol class of $C^{\infty}(X_a\times X_a')$-functions $p(x_a,\xi_a)$ with
$$\lvert \partial_{x_a}^{\alpha}\partial_{\xi_a}^{\beta}p(x_a,\xi_a)\rvert\leq C_{\alpha,\beta}\langle x_a\rangle^{l-\lvert \alpha\rvert}\langle \xi_a\rangle^m,\ \forall (x_a,\xi_a)\in X_a\times X_a',\ \forall \alpha,\beta\in \mathbb N^{\mathrm{dim}X_a}.$$
The corresponding pseudodifferential operators are defined by
$$(p(x_a,D_a)\psi)(x_a)=(2\pi)^{-\nu \mathrm{dim}X_a}\int\int e^{i(x_a-y_a)\xi_a}p(x_a,\xi_a)\psi(y_a)dy_ad\xi_a.$$

\begin{thm}\label{main2}
Suppose the same assumption as in Theorem \ref{main1}, $p(x_a,\xi_a)\in S^0_0$ and that $\mathrm{supp}\ p\subset\{ (x_a,\xi_a)\vert x_a\cdot\xi_a<(1-\epsilon_1)\lvert x_a\rvert \lvert \xi_a\rvert\}$ for some $\epsilon_1>0$. Then for any $f\in C_0^\infty(\mathbb R^+)$ and any $l,s,s'\in \mathbb R$ such that $0\leq l< s<\min\{s',\rho\}$,
\begin{equation}\label{myeq2.0.0.1}
\langle x_a\rangle^lp(x_a,D_a)\pi_{\alpha}e^{-itH}f(H)\langle x\rangle^{-s'}=O(t^{-s+l}),
\end{equation}
in $\mathcal L(\mathcal H)$ as $t\to \infty$.
\end{thm}

\begin{thm}[maximal velocity estimate]\label{main4}
Suppose the same assumption as in Theorem \ref{main1}. Then for any $f\in C_0^\infty(\mathbb R)$ supported in a sufficiently small neighborhood of $E$ and any $l,s,s'\in \mathbb R$ such that $0\leq l< s<\min\{s',\rho\}$,
\begin{equation}\label{myeq2.0.0.2}
\langle x_a\rangle^l\chi\left(E-\lambda_{\alpha}-\frac{x_a^2}{4t^2}<-\epsilon\right)\pi_{\alpha}e^{-itH}f(H)\langle x\rangle^{-s'}=O(t^{-s+l}),
\end{equation}
in $\mathcal L(\mathcal H)$ as $t\to \infty$.
\end{thm}

Let $\pi$ be a projection corresponding to a finite set $\sigma$ of isolated eigenvalues and $\Phi\in C_0^{\infty}(\mathbb R)$ be a function such that $\Phi=1$ near $E$. Then we can construct the projection $\Pi$ onto the almost invariant subspace close to $\mathrm{Ran}\, \pi$.

\begin{thm}\label{projection}
Suppose the same assumption as in Theorem \ref{main1}. Then there exists a projection $\Pi$ satisfying the following conditions.
\begin{itemize}
\item[(1)]For any $r,r',u,u' \in \mathbb R$ and any $f'\in C_0^{\infty}(\mathbb R)$ such that $f'\Phi=f'$, there exists $C_{r,r',u,u'}>0$ such that
\begin{equation}\label{myeq2.0}
\begin{split}
\lVert\langle D\rangle^u\langle x\rangle^rf'(H)&[\Pi,H]\langle x\rangle^{r'}\langle D\rangle^{u'}\rVert_{\mathcal L(\mathcal H)}\\
&+\lVert\langle D\rangle^u\langle x\rangle^r[\Pi,H]f'(H)\langle x\rangle^{r'}\langle D\rangle^{u'}\rVert_{\mathcal L(\mathcal H)}<C_{r,r',u,u'},
\end{split}
\end{equation}
\item[(2)] For any $r,r'\in \mathbb R$ satisfying $r+r'<1+\rho$ there exists $C_{r,r'}'>0$ such that
\begin{equation}\label{myeq2.1}
\lVert \langle x\rangle^r(\Pi-\pi)\langle x\rangle^{r'}\rVert_{\mathcal L(\mathcal H)}<C_{r,r'}'.
\end{equation}
\item[(3)] For any $r,r'\in \mathbb R$ there exists $\tilde C_{r,r'}>0$ such that
\begin{equation}\label{myeq2.1.1}
\lVert \langle x^a\rangle^r\Pi\langle x^a\rangle^{r'}\rVert_{\mathcal L(\mathcal H)}<\tilde C_{r,r'}.
\end{equation}
\end{itemize}
\end{thm}

For $\Pi$ we can remove the restriction $s<\rho$.

\begin{thm}\label{ais1}
Assume the conditions in Theorem \ref{main2}. Let $\pi$ be a projection corresponding to a finite set $\sigma$ of isolated eigenvalues and $\Pi$ be as above. Then \eqref{myeq2.0.1}, \eqref{myeq2.0.0.1} and \eqref{myeq2.0.0.2} hold with $\pi_{\alpha}$ and $s<\min\{s',\rho\}$ replaced by $\Pi$ and $s<s'$ respectively.
\end{thm}

\begin{rem}
Using Theorem \ref{ais1} we can easily see that \eqref{myeq2.0.1}, \eqref{myeq2.0.0.1} and \eqref{myeq2.0.0.2} hold with $\pi_{\alpha}$ and $s<\min\{s',\rho\}$ replaced by $\pi$ and $s<\min\{s',\rho+1\}$ respectively. However, in Theorem \ref{main1}, Theorem \ref{main2} and Theorem \ref{main4} $\pi_{\alpha}$ can be a projection onto an eigenfunction of a degenerated eigenvalue. Therefore we cannot obtain Theorem \ref{main1}, Theorem \ref{main2} and Theorem \ref{main4} from Theorem \ref{ais1}.
\end{rem}

As an application of Theorem \ref{ais1} we obtain the following three-body propagation estimate.

\begin{thm}\label{three}
Assume  the conditions in Theorem \ref{main2} and $N=3$. Then for any $E>\lambda_{\alpha},\ E\notin \mathcal T$, $\epsilon>0$, any $f\in C_0^{\infty}(\mathbb R)$ supported in a sufficiently small neighborhood of $E$ and any $l,s,s'\in \mathbb R$ such that $0\leq l< s<s'$, the following estimate holds.
\begin{equation}\label{myeq2.1.2}
J_a(x)p(x_a,D_a)e^{-itH}f(H)\langle x\rangle^{-s'}=O(t^{-s}),
\end{equation}
where $J_a(x)$ is supported in $\Omega_a:=\{ x\vert \lvert x_a\rvert>\zeta\lvert x\rvert\}\cup \{\lvert x\rvert<\tilde\zeta\}$ for some constants $\zeta,\tilde{\zeta}>0$.
\end{thm}

Let $R(z):=(H-z)^{-1}$ be the resolvent of $H$. From Theorem \ref{main2} we also obtain a resolvent estimate.

\begin{thm}\label{main3}
Suppose the same assumption as in Theorem \ref{main2} and $\rho>m$, $m\in \mathbb N\setminus \{0\}$. Let $I\subset\mathbb R^+$ be a compact interval. Then for all $s',>m$ and $0\leq s<\min\{s',\rho\}-m$ there exists a constant $C>0$ such that
\begin{equation}\label{myeq2.3}
\sup_{\substack{\lambda\in I\\ \mu>0}}\lVert\langle x_a\rangle^s p(x_a,D_a)\pi_{\alpha}(R(\lambda+i\mu))^m\langle x\rangle^{-s'}\rVert<C,
\end{equation}
and for $\lambda \in I$ the limit
\begin{equation}\label{myeq2.4}
\lim_{\substack{\mu\to +0}}\langle x_a\rangle^s p(x_a,D_a)\pi_{\alpha}(R(\lambda+i\mu))^m\langle x\rangle^{-s'},
\end{equation}
exists in $\mathcal L(\mathcal H)$ and define a continuous $\mathcal L(\mathcal H)$-valued function in $\lambda\in \mathbb R^+$.
\end{thm}

\begin{rem}
The corresponding resolvent estimates for $R(\lambda-i\mu)$ and $p$ such that $\mathrm{supp}\ p\subset\{ (x_a,\xi_a)\vert x_a\cdot\xi_a>(-1+\epsilon_1)\lvert x_a\rvert \lvert \xi_a\rvert\}$ can be obtained from the propagation estimates  as $t\to -\infty$ corresponding to Theorem \ref{main2} which can be obtained from other propagation estimates as $t\to -\infty$.
\end{rem}

\begin{rem}\label{comparison}
As already mentioned in the introduction, Theorem \ref{main3} is included in Wang \cite[Corollary 2.7]{Wa2}. Note that since the cutoff function $J_a$ in \cite{Wa2} can be chosen so that the eigenfunction of the channel $\alpha$ decays exponentially on $\mathrm{supp}\ (1-J_a)$, we can see by $\lvert x\rvert^2=\lvert x_a\rvert^2+\lvert x^a\rvert^2$ that $J_a$ is not needed. However, in fact the proof of \cite[Corollary 2.7]{Wa2} is incorrect and we have the restriction $s<\epsilon_0+1/2$ in Corollary 2.7 since the inductive argument in \cite[Theorem 2.5]{Wa2} does not seem to work in the proof of \cite[Corollary 2.7]{Wa2}. This is because the commutator $[\pi,I_a]$ in \cite[(2.24)]{Wa2} cannot be written as
$$[\pi,I_a]=O(\langle x\rangle^{-1-\epsilon_0})\pi+O(\langle x\rangle^{-r})$$
(cf. Remark \ref{critical term}) where $\pi=\pi_{\alpha}$, $I_a$ corresponds to $\tilde V_a=\sum_{(ij)\nless a}V_{ij}(x_i-x_j)$, $r>1/2+s$ and $\epsilon_0$ is the index of the decay of potentials which corresponds to $\rho$ in our assumption (note that we can not insert $\pi$ as in the proof of \cite[Thorem 2.5 ]{Wa2}since we do not have the assumption $a\in \mathcal A^+_2$ in \cite[Corollary 2.7]{Wa2} as in \cite[Theorem 2.5]{Wa2}).

To obtain the propagation estimate as Theorem \ref{main2} with $n-1<s\leq n, n\in \mathbb N\setminus \{0\}$ from the resolvent estimate, we need the estimate for $(R(\lambda+i\mu))^{n+1}$ (see e.g. Perry \cite{Pe}). Since the index $s'$ in the estimate
$$\lVert \pi_{\alpha}b^a_{\pm}(x_a,D_a)(R(\lambda+i\kappa))^{n+1}\langle x\rangle^{-s'}\rVert \leq C\langle \lambda\rangle^{-n/2}$$
needs to satisfy $n+1/2<s'<\epsilon_0+1/2$, $\epsilon_0$ needs to satisfy $n<\epsilon_0$. To sum up,  the index $s$ in the propagation estimate as Theorem \ref{main2} needs to satisfy $s\leq n$ for $n\in \mathbb N$ such that $n<\epsilon_0$. Since $\epsilon_0$ corresponds to $\rho$ in our assumption, this means that our estimate Theorem \ref{main2} is better than that obtained from \cite[Corollary 2.7]{Wa2} with respect to the indices $s$. In particular, only our method gives the propagation estimates for long-range potentials.
\end{rem}

We also obtain improved resolvent estimates using the projection $\Pi$ onto almost invariant subspace.

\begin{thm}\label{ais2}
Assume the conditions in Theorem \ref{main2}. Let $\pi$ and $\Pi$ as in Theorem \ref{ais1}. Then the results in Theorem \ref{main3} hold with $\pi_{\alpha}$, $s'>m$ and $s<s'-m$ replaced by $\Pi$, $s'>m$ and $s<s'-m$ respectively.
\end{thm}

\section{Proof of Theorem \ref{main1}}\label{thirdsec}
To prove Theorem \ref{main1} we need the propagation estimate for the generator of the dilation in $X_a$. This estimate is obtained by  inductive arguments and the positivity of the commutator of a certain operator and the Hamiltonian except for the forms which are integrable with respect to $t$ when applied to $e^{-itH}f(H)\langle x\rangle^{-s'}\phi$, where $\phi\in \mathcal H$. The key estimates are those in \eqref{myeq3.7.1}. Using this propagation estimate Theorem \ref{main1} is proved in a similar way.

We use the following lemma throughout this paper.
\begin{lem}\label{2cldecay}
Let $a$ be a cluster decomposition such that $\# a=2$, and $(ij)\nleq a$. Then for all $s\geq 0$, $\langle x\rangle^s\langle x_{ij}\rangle^{-s}\langle x^a\rangle^{-s}$ is a bounded operator, where $x_{ij}=x_i-x_j$.
\end{lem}
\begin{proof}
We have
\begin{align*}
\langle x\rangle^s\langle x_{ij}\rangle^{-s}\langle x^a\rangle^{-s}\leq& F(\lvert x_a\rvert< C\lvert x^a\rvert)\langle x\rangle^s\langle x_{ij}\rangle^{-s}\langle x^a\rangle^{-s}\\
&+F(\lvert x_a\rvert\geq C\lvert x^a\rvert)\langle x\rangle^s\langle x_{ij}\rangle^{-s}\langle x^a\rangle^{-s}
\end{align*}
Here $F(\dotsm)$ is a characteristic function of the set $\{x\vert \dotsm\}$. By $\lvert x_a\rvert^2+\lvert x^a\rvert^2=\lvert x\rvert^2$ the first term is bounded. 

As for the second term, since $\#a=2$, there exist constants $C_1,C_2>0$ such that $\lvert x_{ij}\rvert\geq C_1\lvert x\rvert-C_2\lvert x^a\rvert$.
 Thus if $C>0$ is large enough, we have
\begin{align*}
F(\lvert x_a\rvert\geq C\lvert x^a\rvert)\langle x\rangle^s&\langle x_{ij}\rangle^{-s}\langle x^a\rangle^{-s}\\
&\leq F(\lvert x_{ij}\rvert\geq C_1\lvert x\rvert-C_2C^{-1}\lvert x_a\rvert)\langle x\rangle^s\langle x_{ij}\rangle^{-s}\langle x^a\rangle^{-s}\\
&\leq F(\lvert x_{ij}\rvert\geq C'\lvert x\rvert)\langle x\rangle^s\langle x_{ij}\rangle^{-s}\langle x^a\rangle^{-s}
<\tilde C,
\end{align*}
where $C'$ and $\tilde C$ are positive constants. Thus the second term is also bounded, so the lemma is proved.
\end{proof}

We also use the following result which is proved in Skibsted \cite{Sk} and can be applied under Assumption \ref{potentialas}.
\begin{lem}[Skibsted \cite{Sk}]\label{Skibstedlem}
Let $E>\inf\sigma_{\mathrm{ess}}(H),\ E\notin \mathcal T$ be given. Then for any $f\in C_0^{\infty}(\mathbb R)$ supported in a sufficiently small neighborhood of $E$ and any $s'>s>0$,
$$\chi\left(\frac{x^2}{4t^2}-d(E)<-\epsilon\right)e^{-itH}f(H)\langle x\rangle^{-s'}=O(t^{-s}),$$
in $\mathcal L(\mathcal H)$ as $t\to \infty$ where $d(E):=\mathrm{dist}(E,\mathcal T)$.
\end{lem}

\begin{rem}
In Skibsted \cite{Sk} we assume $E>0$ for Lemma \ref{Skibstedlem}. However, we can easily modify the proof so that the lemma holds for $E\notin \mathcal T,E>\inf\sigma_{\mathrm{ess}}(H)$ using the Mourre estimate \cite[Theorem B2]{Sk}.
\end{rem}

Without loss of generality we can assume that intercluster potential $\tilde V_a$ is smooth. Let $R$ be as in Assumption \ref{potentialas} and $F\in C_0^{\infty}(\mathbb R^{\nu})$ be a function such that $F(y)=1$ for $\lvert y\rvert<R$. Set $\hat V_a(x):=\sum_{(ij)\notin a}V_{ij}(x_i-x_j)(1-F(x_i-x_j))$. Then $\hat V_a(x)\in C^{\infty}(X)$ and we can rewrite
\begin{equation}\label{myeq3.0.1}
\begin{split}
\pi_{\alpha}&e^{-itH}f(H)\langle x\rangle^{-s'}\\
&=\pi_{\alpha}f'(H')e^{-itH}f(H)\langle x\rangle^{-s'}+\pi_{\alpha}(f'(H)-f'(H'))e^{-itH}f(H)\langle x\rangle^{-s'},
\end{split}
\end{equation}
where $H':=H^a-\Delta_a+\hat V_a$ and $f'\in C_0^{\infty}(\mathbb R)$ satisfies $ff'=f$. Since $\pi_{\alpha}(f'(H)-f'(H'))\langle x\rangle^n\in \mathcal L(\mathcal H)$ for any $n\in \mathbb N$ by Lemma \ref{2cldecay} and the exponential decay of the eigenfunctions, the second term in \eqref{myeq3.0.1} is estimated as $O(t^{-s})$ by Lemma \ref{Skibstedlem}. Since 
\begin{align*}
\pi_{\alpha}f'(H')e^{-itH}f(H)\langle x\rangle^{-s'}=&\pi_{\alpha}f'(H')\pi_{\alpha}e^{-itH}f(H)\langle x\rangle^{-s'}\\
&+\pi_{\alpha}[\pi_{\alpha},f'(H')]e^{-itH}f(H)\langle x\rangle^{-s'},
\end{align*}
and $[\pi_{\alpha},f'(H')]\langle x\rangle^{1+\rho}\in \mathcal L(\mathcal H)$, we only need to prove
$$\chi\left (\frac{x_a^2}{4t^2}-(E-\lambda_{\alpha})<-\epsilon\right)\pi_{\alpha}f'(H')\pi_{\alpha}e^{-itH}f(H)\langle x\rangle^{-s'}=O(t^{-s}).$$
However, we have
\begin{equation}\label{myeq3.0.2}
\begin{split}
f'(H')\pi_{\alpha}e^{-itH}f(H)\langle x\rangle^{-s'}=&f'(H')e^{-itH'}\pi_{\alpha}f(H)\langle x\rangle^{-s'}\\
&+\int_0^tf'(H')e^{-i(t-r)H'}(H'\pi_{\alpha}-\pi_{\alpha}H)e^{-irH}f(H)\langle x\rangle^{-s'}.
\end{split}
\end{equation}
As for the second term in \eqref{myeq3.0.2} we have $\langle x\rangle^{r}(H'\pi_{\alpha}-\pi_{\alpha}H)\langle x\rangle^{r'}\in \mathcal L(\mathcal H)$ for any $r,r'\in \mathbb R$ such that $r+r'<1+\rho$. We decompose the integral into two pieces corresponding to $[0,t/2]$ and $[t/2,t]$. As for the first part, if we assume Theorem \ref{main1} holds with $e^{-itH}f(H)$ replaced by $e^{-itH'}f'(H')$, we can estimate as
$$\chi\left (\frac{x_a^2}{4t^2}-(E-\lambda_{\alpha})<-\epsilon\right)\pi_{\alpha}f'(H')e^{-i(t-r)H'}\langle x\rangle^{-1-\rho}=O((t-r)^{-1-\rho})=O(t^{-1-\rho}).$$
As for the second part, by Lemma \ref{Skibstedlem} and that for any $\epsilon>0$
$$\int_0^{t/2}\langle t-r\rangle^{-1-\epsilon}dr<C$$
with $C$ independent of $t$, we can estimate the part as $O(t^{-s})$ for any $s<\rho$. Therefore, the second term in \eqref{myeq3.0.2} multiplied by $\chi\left (\frac{x_a^2}{4t^2}-(E-\lambda_{\alpha})<-\epsilon\right)\pi_{\alpha}$ is estimated as $O(t^{-s})$ for any $s<\rho$. Assuming Theorem \ref{main1} holds with $e^{-itH}f(H)$ replaced by $e^{-itH'}f'(H')$, the first term  in \eqref{myeq3.0.2} multiplied by
$$\chi\left (\frac{x_a^2}{4t^2}-(E-\lambda_{\alpha})<-\epsilon\right)\pi_{\alpha},$$
is estimated as $O(t^{-s})$ for any $s<\rho$. Thus, we only need to prove Theorem \ref{main1} replacing $\tilde V_a$ by $\hat V_a$. That is, we can assume $\tilde V_a$ is smooth.

We introduce a class of functions which is introduced in Skibsted \cite{Sk}.
\begin{defin}[Skibsted \cite{Sk}]
Given $\beta,\gamma\geq0$ and $\epsilon>0$, let $\mathcal F_{\beta,\gamma,\epsilon}$ denote the set of functions $g$, $g(x,\tau)=g_{\beta,\gamma,\epsilon}(x,\tau)=-\tau^{-\beta}(-x)^{\gamma}\chi_{\epsilon}\left(\frac{x}{\tau}\right)$, defined for $(x,\tau)\in \mathbb R\times \mathbb R^+$ and for $\chi_{\epsilon}\in C^{\infty}(\mathbb R)$ with the following properties:
\begin{align*}
&\chi_{\epsilon}(x)=1\ \mathrm{for}\ x<-2\epsilon,\ \chi_{\epsilon}(x)=0\ \mathrm{for}\ x>-\epsilon,\\
&\frac{d}{dx}\chi_{\epsilon}(x)\leq 0\ \mathrm{and}\ \alpha\chi_{\epsilon}(x)+x\frac{d}{dx}\chi_{\epsilon}(x)=\tilde \chi^2_{\epsilon}(x),
\end{align*}
where $\tilde \chi_{\epsilon}(x)\geq0$ and $\tilde \chi_{\epsilon}\in C^{\infty}(\mathbb R)$.
\end{defin}

Set $t_0>0$ and $E'<E$. We define $A(\tau)=\frac{1}{2}(x_ap_a+p_ax_a)-2(E'-\lambda_{\alpha})\tau$ for $\tau=t+t_0,\ t\geq0$, where $p_a=-i\nabla_a$. To prove the Theorem \ref{main1} first, we prove the following lemma which corresponds to Skibsted \cite[Example1]{Sk}.

\begin{lem}\label{propagation1}
Suppose the same assumption as in Theorem \ref{main1}. Let $E>\lambda_{\alpha},\ E\notin \mathcal T$ and $\lambda_{\alpha}<E'<E$ be given. Then for any $f\in C_0^{\infty}(\mathbb R)$ supported in a sufficiently small neighborhood of $E$, any $\epsilon>0$ and any $s,s',\tilde s\in \mathbb R$ such that $0\leq\tilde s\leq s< \min\{s',\rho\}$,
$$\left(\frac{-A(\tau)}{\tau}\right)^{\tilde s}\chi\left(\frac{A(\tau)}{\tau}<-\epsilon\right)\pi_{\alpha}e^{-itH}f(H)\langle x\rangle^{-s'}=O(t^{-s})$$
in $\mathcal L(\mathcal H)$ as $t\to \infty$.
\end{lem}

\begin{proof}
We may suppose $s'<\rho$, since if $s<\rho\leq s'$, there exists $s_0'$ such that $s<s_0'<\rho$ and we have $\langle x\rangle^{-s'}=\langle x\rangle^{-s'_0}\langle x\rangle^{-(s'-s_0')}$.
Since for any $\tilde \gamma\leq \gamma$ we have
\begin{equation}\label{myeq3.0}
-g_{\beta,\gamma,\epsilon}(x,\tau)\geq -\tau^{-\beta}(\epsilon \tau)^{\gamma-\tilde \gamma}g_{0,\tilde \gamma,\epsilon}(x,\tau),
\end{equation}
we only need to prove that $\lVert(-g_{\beta_0,\gamma_0,\epsilon}(A(\tau),\tau))^{1/2}\pi_{\alpha}\psi(t)\rVert\leq C\lVert \phi\rVert$ for any $\beta_0>0$, $\epsilon>0$ and $\gamma_0,\gamma'_0\in \mathbb R$ such that $0< \gamma_0<\gamma'_0<2\rho$ and $\phi\in \mathcal H$, where $\psi(t)=\psi_{\gamma'_0}(t)=e^{-itH}f(H)\langle x\rangle^{-\gamma'_0/2}\phi$.

Let $\epsilon$ be a positive number such that $64\epsilon<d(E)$. Set $\theta_{\epsilon}:=\eta_{\epsilon}(\frac{x^2}{4\tau^2})$, where $\eta_{\epsilon}(x)$ is a $C^{\infty}(\mathbb R)$ function such that $\eta_{\epsilon}(x)=1$ for $x<\epsilon$ and $\eta_{\epsilon}(x)=0$ for $x>2\epsilon$. We also set $\tilde \theta_{\epsilon}=1-\theta_{\epsilon}$ and $g=g_{\beta_0,\gamma_0,\epsilon}(A(\tau),\tau)$.
Then we have
\begin{equation}\label{myeq3.1}
\begin{split}
(\pi_{\alpha}\psi,g\pi_{\alpha}\psi)
=&(\pi_{\alpha}\theta_{\epsilon}\psi,g\pi_{\alpha}\theta_{\epsilon}\psi)+(\pi_{\alpha}\tilde\theta_{\epsilon}\psi,g\pi_{\alpha}\theta_{\epsilon}\psi)\\
&+(\pi_{\alpha}\theta_{\epsilon}\psi,g\pi_{\alpha}\tilde\theta_{\epsilon}\psi)+(\pi_{\alpha}\tilde\theta_{\epsilon}\psi,g\pi_{\alpha}\tilde\theta_{\epsilon}\psi).
\end{split}
\end{equation}
As for the first term of \eqref{myeq3.1} we have
\begin{equation}\label{myeq3.1.0}
g\pi_{\alpha}\theta_{\epsilon}\psi=g\pi_{\alpha}\theta_{\epsilon} f_1(H)\theta_{2\epsilon}\psi+g\pi_{\alpha}\theta_{\epsilon} f_1(H)\tilde \theta_{2\epsilon}\psi,
\end{equation}
where $f_1\in C_0^{\infty}(\mathbb R)$ such that $f_1f=f$. We can easily see $\lVert g\pi_{\alpha}\theta_{\epsilon} f_1(H)\tilde \theta_{2\epsilon}\rVert_{\mathcal L(\mathcal H)}=O(\tau^{-\infty})$. Since observing 
\begin{equation}\label{myeq3.7.0.2}
\begin{split}
&\lVert \langle A(\tau)\rangle^r(\langle p_ax_a+x_ap_a\rangle^r +\tau^r+1)^{-1}\rVert_{\mathcal L(\mathcal H)}=O(1)\ \mathrm{as}\ \tau\to \infty, \ \forall r>0,\\
&\langle p_ax_a+x_ap_a\rangle^r\langle x_a\rangle^{-r}\langle p_a\rangle^{-r}\in \mathcal L(\mathcal H),\ \forall r>0,
\end{split}
\end{equation}
it is easily seen that $g\pi_{\alpha}\theta_{\epsilon} f_1(H)=O(\tau^{\gamma_0-\beta_0})$, by Lemma \ref{Skibstedlem} we have
$$(\pi_{\alpha}\theta_{\epsilon}f_1(H) \theta_{2\epsilon}\psi,g\pi_{\alpha}\theta_{\epsilon}f_1(H) \theta_{2\epsilon}\psi)\leq C\lVert \phi\rVert^2.$$
Thus the first term of \eqref{myeq3.1} is uniformly bounded with respect to $t$. 

As for the second term we have by Taylor's formula  we have
\begin{equation}\label{myeq3.1.1}
\begin{split}
\pi_{\alpha}\theta_{2\epsilon}&=\pi_{\alpha}\eta_{2\epsilon}\left(\frac{\lvert x_a\rvert^2}{4\tau^2}\right)+\pi_{\alpha}\int_0^1\eta_{2\epsilon}^{(1)}\left(\frac{\lvert x_a\rvert^2+r\lvert x^a\rvert^2}{4\tau^2}\right)\frac{\lvert x^a\rvert^{2}}{4\tau^2}dr\\
&=\theta_{2\epsilon}\pi_{\alpha}-\int_0^1\eta_{2\epsilon}^{(1)}\left(\frac{\lvert x_a\rvert^2+r\lvert x^a\rvert^2}{4\tau^2}\right)\frac{\lvert x^a\rvert^{2}}{4\tau^2}dr\pi_{\alpha}\\
&\quad +\pi_{\alpha}\int_0^1\eta_{2\epsilon}^{(1)}\left(\frac{\lvert x_a\rvert^2+r\lvert x^a\rvert^2}{4\tau^2}\right)\frac{\lvert x^a\rvert^{2}}{4\tau^2}dr.
\end{split}
\end{equation}
Let us notice
\begin{equation}\label{myeq3.1.2}
\eta_{2\epsilon}^{(1)}\left(\frac{\lvert x_a\rvert^2+r\lvert x^a\rvert^2}{4\tau^2}\right)\theta_{\epsilon}=0\ (r\in [0,1]).
\end{equation}
As for the second term in \eqref{myeq3.1.1} we have
\begin{equation}\label{myeq3.1.3}
\begin{split}
\eta_{2\epsilon}^{(1)}\left(\frac{\lvert x_a\rvert^2+r\lvert x^a\rvert^2}{4\tau^2}\right)=&F(\lvert x^a\rvert\leq \lvert x_a\rvert)\eta_{2\epsilon}^{(1)}\left(\frac{\lvert x_a\rvert^2+r\lvert x^a\rvert^2}{4\tau^2}\right)\\
&+F(\lvert x_a\rvert< \lvert x^a\rvert)\eta_{2\epsilon}^{(1)}\left(\frac{\lvert x_a\rvert^2+r\lvert x^a\rvert^2}{4\tau^2}\right),
\end{split}
\end{equation}
where $F(\lvert x^a\rvert\leq \lvert x_a\rvert)$ and $F(\lvert x_a\rvert< \lvert x^a\rvert)$ are indicator functions for $\{x\vert \lvert x^a\rvert\leq \lvert x_a\rvert\}$ and $\{x\vert \lvert x_a\rvert< \lvert x^a\rvert\}$ respectively.
Since we have
$$\frac{\lvert x_a\rvert^2+\lvert x^a\rvert^2}{4\tau^2}\leq \frac{2\lvert x_a\rvert^2}{4\tau^2}\leq 8\epsilon$$
for $x\in \mathrm{supp}\,  \eta_{2\epsilon}^{(1)}\left(\frac{\lvert x_a\rvert^2+r\lvert x^a\rvert^2}{4\tau^2}\right)\cap \{x\vert \lvert x^a\rvert\leq \lvert x_a\rvert\},$
we obtain 
\begin{equation}\label{myeq3.1.4}
F( \lvert x^a\rvert\leq \lvert x_a\rvert)\eta_{2\epsilon}^{(1)}\left(\frac{\lvert x_a\rvert^2+r\lvert x^a\rvert^2}{4\tau^2}\right)=\theta_{8\epsilon}F(\lvert x^a\rvert\leq \lvert x_a\rvert)\eta_{2\epsilon}^{(1)}\left(\frac{\lvert x_a\rvert^2+r\lvert x^a\rvert^2}{4\tau^2}\right).
\end{equation}
Moreover, since by the exponential decay of the eigenfunction we have $\langle x^a\rangle^n\pi_{\alpha}\langle x^a\rangle^m\in \mathcal L(\mathcal H)$ for any $m, n>0$, we obtain
\begin{equation}\label{myeq3.1.5}
F(\lvert x_a\rvert< \lvert x^a\rvert)\eta_{2\epsilon}^{(1)}\left(\frac{\lvert x_a\rvert^2+r\lvert x^a\rvert^2}{4\tau^2}\right)\pi_{\alpha}=O(\tau^{-\infty}).
\end{equation}
By \eqref{myeq3.1.1}-\eqref{myeq3.1.5} and $\theta_{\epsilon}=\theta_{2\epsilon}\theta_{\epsilon}$ the second term in \eqref{myeq3.1} can be written as
\begin{equation}\label{myeq3.1.6}
\begin{split}
(\pi_{\alpha}\tilde\theta_{\epsilon}\psi,g\pi_{\alpha}\theta_{\epsilon}\psi)=\Bigg(\pi_{\alpha}\tilde\theta_{\epsilon}&\psi,g\theta_{8\epsilon}F(\lvert x^a\rvert\leq \lvert x_a\rvert)\int_0^1\eta_{2\epsilon}^{(1)}\left(\frac{\lvert x_a\rvert^2+r\lvert x^a\rvert^2}{4\tau^2}\right)\\
&\cdot\frac{\lvert x^a\rvert^{2}}{4\tau^2}dr\pi_{\alpha}\theta_{\epsilon}\psi\Bigg)+O(\tau^{-\infty})
\end{split}
\end{equation}

Noting $\theta_{8\epsilon}=\theta_{16\epsilon}\theta_{8\epsilon}$ and $\left\lVert\mathrm{ad}_{A(\tau)}^n(\theta_{16\epsilon})\right\rVert_{\mathcal L(\mathcal H)}=O(1)$ and using an almost analytic extension we have
$$g\theta_{8\epsilon}=\theta_{16\epsilon}\sum_{n=0}^{[\gamma_0]}\mathrm{ad}_{\theta_{16\epsilon}}^n(g)\theta_{8\epsilon}+\mathrm{ad}_{\theta_{16\epsilon}}^{[\gamma_0]+1}(g)\theta_{8\epsilon},$$
where $\lVert\mathrm{ad}_{\theta_{16\epsilon}}^{[\gamma_0]+1}(g)\theta_{8\epsilon}\rVert_{\mathcal L(\mathcal H)}=O(1)$ (to justify this calculation we need to use an approximation argument for $g$ as in the proof of Lemma \ref{timederi}). Finally, noting $\theta_{16\epsilon}=\theta_{32\epsilon}\theta_{16\epsilon}$ and commuting $\pi_{\alpha}$ and $\theta_{32\epsilon}$, by the similar argument as in the estimate of the first term of \eqref{myeq3.1}, the right-hand side of \eqref{myeq3.1.6} is bounded by $C\lVert \phi\rVert^2$.
 In the same way it is seen that the third term of \eqref{myeq3.1} is uniformly bounded with respect to $t$.
 
It remains to estimate the fourth term of \eqref{myeq3.1}. In the following lemma we use the functions $b_{\gamma,\epsilon}(x):=(-x)^{\gamma}\chi_{\epsilon}(x),\ c_{\gamma,\epsilon}(x):=(-x)^{(\gamma-1)/2}\tilde\chi_{\epsilon}(x)$ and their almost analytic extensions (see Adachi \cite{Ad}).
\begin{lem}\label{timederi}
Set $\tilde\theta=\tilde\theta_{\epsilon},\ g=g_{\beta,\gamma,\epsilon}$, $\psi(t)=e^{-itH}f(H)\langle x\rangle^{-\gamma'/2}\phi$, where $\epsilon>0$, $\beta\geq 0$ and $2\rho>\gamma'\geq \gamma>0$. Then with the convention $\langle P\rangle_t=\langle\psi(t),P\psi(t)\rangle$ for an operator or a form $P$, $\langle \tilde \theta\pi_{\alpha}g(A(\tau),\tau)\pi_{\alpha}\tilde\theta\rangle_t$ is absolutely continuous with $\frac{d}{dt}\langle \tilde \theta\pi_{\alpha}g(A(\tau),\tau)\pi_{\alpha}\tilde\theta\rangle_t=\langle \sum_{i=1}^{14}E_i\rangle_t$, where $E_i$ is given as follows:
\begin{align*}
&E_1=\tilde \theta\pi_{\alpha}\left(\frac{\partial}{\partial\tau}g\right)(A(\tau),\tau)\pi_{\alpha}\tilde\theta,\\
&E_2=2\mathrm{Re}(\tilde\theta i[f_1(H)H,\pi_{\alpha}]g(A(\tau),\tau)\pi_{\alpha}\tilde\theta),\\
&E_3=2\mathrm{Re}(D_1\tilde\theta \pi_{\alpha}g(A(\tau),\tau)\pi_{\alpha}\tilde\theta),
\end{align*}
with $D_1B(\tau)=i\mathrm{ad}_{B(\tau)}(f_1(H)H)+d_{\tau}B(\tau)$,
\begin{equation*}
E_4=\tilde\theta\pi_{\alpha}\tilde g(A(\tau),\tau)f_2(H)DA(\tau)f_2(H)\tilde g(A(\tau),\tau)\pi_{\alpha}\tilde\theta,
\end{equation*}
with $DB(\tau)=i\mathrm{ad}_{B(\tau)}(H)+d_{\tau}B(\tau)$, $\tilde g=(g^{(1)})^{1/2}$ and $f_2\in C_0^{\infty}(\mathbb R)$ satisfies $f_1f_2=f_1$, 
\begin{align*}
&E_5=\tilde\theta\pi_{\alpha}(1-f_2(H))\tilde g(A(\tau),\tau)D_1A(\tau)f_2(H)\tilde g(A(\tau),\tau)\pi_{\alpha}\tilde\theta,\\
&E_6=\tilde\theta\pi_{\alpha}\sum_{m=1}^{k}(m!)^{-1}\tilde g^{(m)}(A(\tau),\tau)\mathrm{ad}_{A(\tau)}^m(f_2(H))D_1A(\tau)f_2(H)\tilde g(A(\tau),\tau)\pi_{\alpha}\tilde\theta,\\
&E_7=\tilde\theta\pi_{\alpha}R_0(\tau)D_1A(\tau)f_2(H)\tilde g(A(\tau),\tau)\pi_{\alpha}\tilde\theta,
\end{align*}
with 
\begin{equation}\label{myeq3.1.8}
k=
\begin{cases}
[\gamma]/2+1 &\mathrm{if}\ [\gamma]+1\ \mathrm{is\ odd}\\
([\gamma]+1)/2 &\mathrm{if}\ [\gamma]+1\ \mathrm{is\ even}
\end{cases}
\end{equation}
and
$$R_0(\tau)=\frac{(-1)^k}{2\pi i}\tau^{(\gamma-1-\beta)/2-k-1}\int_{\mathbb C}\bar \partial_z\tilde  c(z)(\hat A-z)^{-k-1}\mathrm{ad}_{A(\tau)}^{k+1}(f_2(H))(\hat A-z)^{-1}dz\wedge d\bar z,$$
where $\tilde c$ is an almost analytic extension of $c_{\gamma,\epsilon}$ and $\hat A=\tau^{-1}A(\tau)$.
\begin{align*}
&E_8=\tilde\theta\pi_{\alpha}\tilde g(A(\tau),\tau)\sum_{m=1}^{k}\frac{(-1)^m}{m!}\mathrm{ad}_{A(\tau)}^m(D_1A(\tau)f_2(H))\tilde g^{(m)}(A(\tau),\tau)\pi_{\alpha}\tilde\theta,\\
&E_9=\tilde\theta\pi_{\alpha}\tilde g(A(\tau),\tau)R_1(\tau)\pi_{\alpha}\tilde\theta
\end{align*}
with $R_1(\tau)=-\frac{1}{2\pi i}\tau^{(\gamma-1-\beta)/2-k-1}\int_{\mathbb C}\bar \partial_z\tilde c(z)(\hat A-z)^{-1}\mathrm{ad}_{A(\tau)}^{k+1}(D_1A(\tau)f_2(H))(\hat A-z)^{-k-1}dz\wedge d\bar z$.
\begin{align*}
&E_{10}=\tilde\theta\pi_{\alpha}\tilde g(A(\tau),\tau)\sum_{m=0}^{k}\frac{(-1)^{m}}{m!}\mathrm{ad}_{A(\tau)}^m(D_1A(\tau))\tilde g^{(m)}(A(\tau),\tau)(1-f_2(H))\pi_{\alpha}\tilde\theta,\\
&E_{11}=\tilde\theta\pi_{\alpha}\tilde g(A(\tau),\tau)R_2(\tau)(1-f_2(H))\pi_{\alpha}\tilde\theta
\end{align*}
with $R_2(\tau)=-\frac{1}{2\pi i}\tau^{(\gamma-1-\beta)/2-k-1}\int_{\mathbb C}\bar \partial_z\tilde c(z)(\hat A-z)^{-1}\mathrm{ad}_{A(\tau)}^{k+1}(D_1A(\tau))(\hat A-z)^{-k-1}dz\wedge d\bar z$,
\begin{align*}
E_{12}=&\tilde\theta\pi_{\alpha}\sum_{m=2}^{k_1}g_m(A(\tau),\tau)h_m(A(\tau),\tau)\\
&\cdot\sum_{m_1=0}^{[\gamma]+1-m}\frac{(-1)^{m_1}}{m_1!}\mathrm{ad}_{A(\tau)}^{m+m_1-1}(D_1A(\tau))g_m^{(m_1)}(A(\tau),\tau)\pi_{\alpha}\tilde\theta,
\end{align*}
with $k_1=\max\{2,[\gamma]+1\}$, $g_m=g_{\beta/2,((\gamma-m)/2)_+,\epsilon/2}$ and
$$h_m(x,\tau)=\frac{\tau^{\beta}}{m!}(-x)^{-(\gamma-m)_+}g^{(m)}(x,\tau).$$
\begin{align*}
&E_{13}=\tilde\theta\pi_{\alpha}\sum_{m=2}^{k_1}g_m(A(\tau),\tau)h_m(A(\tau),\tau)R_{m+1}(\tau)\pi_{\alpha}\tilde\theta,
\end{align*}
with
\begin{align*}
R_{m+1}(\tau)=&\frac{1}{2\pi i}\tau^{((\gamma-m)/2)_+-\beta/2-[\gamma]+m-2}\\
&\int_{\mathbb C}\bar \partial_z\tilde b_m(z)(\hat A-z)^{-1}\mathrm{ad}_{A(\tau)}^{[\gamma]+1}(D_1A(\tau))(\hat A-z)^{-[\gamma]+m-2}dz\wedge d\bar z,
\end{align*}
where $\tilde b_m$ is an almost analytic extension of $b_{((\gamma-m)/2)_+,\epsilon/2}$,
\begin{align*}
&E_{14}=\tilde\theta\pi_{\alpha}R'(\tau)\pi_{\alpha}\tilde\theta,
\end{align*}
with $R'(\tau)=\frac{(-1)^{k_1}}{2\pi i}\tau^{\gamma-\beta-k_1-1}\int_{\mathbb C}\bar \partial_z\tilde b(z)(\hat A-z)^{-k_1-1}\mathrm{ad}_{A(\tau)}^{k_1}(D_1A(\tau))(\hat A-z)^{-1}dz\wedge d\bar z$ where $\tilde b$ is an almost analytic extension of $b_{\gamma,\epsilon}$.
\end{lem}
\begin{proof}
We set $g_{\delta}(x,\tau)=g_{\beta,\gamma,\epsilon}(x,\tau)F^2(\delta \tau^{-1}x)$, where $F\in C_0^{\infty}(\mathbb R)$ is real and $F(x)=1$ for $\lvert x\rvert<1$. Since $\mathrm{ad}^m_{A(\tau)}(H)$ is a sum of terms of the form $C\Delta_a$ or $C(x_a\nabla_a)^n\tilde V_a,\ n\in \mathbb N,\ n\leq m$, where $C$ is a constant and $\tilde V_a=\sum_{(ij)\nless a}V_{ij}(x_i-x_j)$, by Lemma \ref{2cldecay} and Assumption \ref{potentialas}, $A(\tau)$ satisfies Assumption \ref{commutatoras2} of the appendix. Therefore Lemma \ref{commutatorlem2} holds for $A(\tau)$.

We set $b_{\delta}(x)=b_{\gamma,\epsilon}(x)F(\delta x)$. Then we can write $g_{\delta}(x,\tau)=-\tau^{\gamma-\beta}b_{\delta}(\tau^{-1}x)$, and use an almost analytic extension of $b_{\delta}$ (see Adachi \cite{Ad}). Since $\lambda_{\alpha}$ is in the discrete spectrum, by the exponential decay of the eigenfunction we have $\langle x^a\rangle^n\pi_{\alpha}\in \mathcal L(\mathcal H)$ for any $n\in \mathbb N$. Thus observing
$$-(\gamma-\beta)\tau^{\gamma-\beta-1}b_{\delta}(\tau^{-1}A(\tau))+\tau^{\gamma-\beta-2}A(\tau)b_{\delta}^{(1)}(\tau^{-1}A(\tau))=\left(\frac{\partial}{\partial\tau}g_{\delta}\right)(A(\tau),\tau),$$
and $-\tau^{\gamma-\beta-m}b_{\delta}^{(m)}(\tau^{-1}A(\tau))=g_{\delta}^{(m)}(A(\tau),\tau)$, by Lemma \ref{commHeSj} for any $\delta,t'>0$,
\begin{equation}\label{myeq3.2}
\langle \tilde \theta\pi_{\alpha}g_{\delta}(A(\tau'),\tau')\pi_{\alpha}\tilde\theta\rangle_{t'}-\langle \tilde \theta\pi_{\alpha}g_{\delta}(A(t_0),t_0)\pi_{\alpha}\tilde\theta\rangle_0
=\int_0^{t'}\langle D\rangle_{t},
\end{equation}
where $\tau'=t'+t_0$ and
\begin{equation}\label{myeq3.3}
\begin{split}
D&=\tilde \theta\pi_{\alpha}\left(\frac{\partial}{\partial\tau}g_{\delta}\right)(A(\tau),\tau)\pi_{\alpha}\tilde\theta+2\mathrm{Re}(\tilde\theta i[f_1(H)H,\pi_{\alpha}]g_{\delta}(A(\tau),\tau)\pi_{\alpha}\tilde\theta)\\
&+2\mathrm{Re}(D_1\tilde\theta \pi_{\alpha}g_{\delta}(A(\tau),\tau)\pi_{\alpha}\tilde\theta)+\sum_{m=1}^{k_1}(m!)^{-1}\tilde\theta\pi_{\alpha}g^{(m)}_{\delta}\mathrm{ad}^{m-1}_{A(\tau)}(D_1A(\tau))\pi_{\alpha}\tilde\theta\\
&+\tilde\theta\pi_{\alpha}\frac{(-1)^{k_1+1}}{2\pi i}\tau^{\gamma-\beta-k_1-1}\int_{\mathbb C}\bar \partial_z\tilde b_{\delta}(z)(\hat A-z)^{-k_1-1}\\
&\qquad\cdot\mathrm{ad}_{A(\tau)}^{k_1}(D_1A(\tau))(\hat A-z)^{-1}dz\wedge d\bar z\pi_{\alpha}\tilde\theta,
\end{split}
\end{equation}
where $\tilde b_{\delta}$ is an almost analytic extension of $b_{\delta}$. We rewrite the right-hand side of \eqref{myeq3.3} as sum of forms $\sum_{l=1}^{14}E_{l,\delta}(\tau)$ such that
\begin{equation}\label{myeq3.4}
\langle E_{l,\delta}(\tau)\rangle_t\to\langle E_{l}(\tau)\rangle_t,\ \mathrm{as}\ \delta\to 0,
\end{equation}
and for any $T>0$
\begin{equation}\label{myeq3.5}
\lvert\langle E_{l,\delta}(\tau)\rangle_t\rvert\leq C_T,\ \forall 0<t<T,\ 0<\delta<1.
\end{equation}

Since by Lemma \ref{communi} the left-hand side of \eqref{myeq3.2} goes to
$$\langle \tilde \theta\pi_{\alpha}g(A(\tau'),\tau')\pi_{\alpha}\tilde\theta\rangle_{t'}-\langle \tilde \theta\pi_{\alpha}g(A(t_0),t_0)\pi_{\alpha}\tilde\theta\rangle_0,$$
as $\delta\to 0$ the lemma follows from the dominated convergence theorem.

The first three terms of \eqref{myeq3.3} are denoted by $E_{1,\delta}$,  $E_{2,\delta}$ and $E_{3,\delta}$ respectively. Then, we can easily see that for $l=1,2,3$, \eqref{myeq3.4} and \eqref{myeq3.5} holds.

We look at the contributions from the fourth term $\tilde\theta\pi_{\alpha}g^{(1)}_{\delta}D_1A(\tau)\pi_{\alpha}\tilde\theta$ on the right hand side of \eqref{myeq3.3}. Using the abbreviations $g=g(A(\tau),\tau)$, $\tilde g=\tilde g(A(\tau),\tau)$, $F=F(\delta \tau^{-1}A(\tau))$ and $F'=\frac{\partial}{\partial x}F(\delta \tau^{-1}x)_{\vert x=A(\tau)}$ we have
\begin{align*}
&\tilde\theta\pi_{\alpha}g^{(1)}_{\delta}D_1A(\tau)\pi_{\alpha}\tilde\theta=D_1+\dotsm +D_5,\\
&D_1=\tilde\theta\pi_{\alpha}2\tau^{-1}FF'g D_1A(\tau)\pi_{\alpha}\tilde\theta,\\
&D_2=\tilde\theta\pi_{\alpha}F\tilde gf_2(H)D_1A(\tau)f_2(H)\tilde gF\pi_{\alpha}\tilde\theta,\\
&D_3=\tilde\theta\pi_{\alpha}F\tilde g(1-f_2(H))D_1A(\tau)f_2(H)\tilde gF\pi_{\alpha}\tilde\theta,\\
&D_4=-\tilde\theta\pi_{\alpha}F\tilde g[D_1A(\tau)f_2(H),\tilde gF]\pi_{\alpha}\tilde\theta,\\
&D_5=\tilde\theta F^2\tilde g^2D_1A(\tau)(1-f_2(H))\pi_{\alpha}\tilde\theta.
\end{align*}

As for $D_1$ it is easy to see that $\langle D_1\rangle_t\to 0$ as $\delta\to 0$ and \eqref{myeq3.5} holds. As for $D_2$ by Lemma \ref{communi} and spectral theorem $\langle D_2\rangle_t\to \langle E_4\rangle_t$ as $\delta\to 0$ and \eqref{myeq3.5} holds. By Lemma \ref{commutatorlem2} and Lemma \ref{commHeSj} we see that $\langle D_3\rangle_t \to \langle E_5\rangle_t+\langle E_6\rangle_t+\langle E_7\rangle_t$ and \eqref{myeq3.5} holds. In the same way we have $\langle D_4\rangle_t \to \langle E_8\rangle_t+\langle E_9\rangle_t$ and $\langle D_5\rangle_t \to \langle E_{10}\rangle_t+\langle E_{11}\rangle_t$. The terms with $m>1$ of the fourth term of \eqref{myeq3.3} contributes with $\langle E_{12}\rangle_t+\langle E_{13}\rangle_t$ in the limit $\delta\to 0$ and the last term of \eqref{myeq3.3} contributes with $\langle E_{14}\rangle_t$.
\end{proof}

Now we shall prove 
\begin{equation}\label{myeq3.6}
(\pi_{\alpha}\tilde\theta_{\epsilon}\psi_{\gamma'_0},g_{\beta_0,\gamma_0,\epsilon}(A(\tau),\tau)\pi_{\alpha}\tilde\theta_{\epsilon}\psi_{\gamma'_0})\leq C\lVert \phi \rVert^2,
\end{equation}
for any $\beta_0>0$, $\epsilon>0$ and $0\leq \gamma_0<\gamma'_0<2\rho$.

Under the assumption of Lemma \ref{timederi} we have for any $t'>0$

\begin{equation}\label{myeq3.6.0}
\langle -\tilde \theta \pi_{\alpha}g_{\beta,\gamma,\epsilon}(A(\tau'), \tau') \pi_{\alpha}\tilde \theta\rangle_{t'}=\langle -\tilde \theta \pi_{\alpha}g_{\beta,\gamma,\epsilon}(A(t_0), t_0) \pi_{\alpha}\tilde \theta\rangle_0-\int_0^{t'}dt\langle E_1+\dotsm+E_{14}\rangle_t,
\end{equation}
where $\tau'=t'+t_0$.

We notice that the following estimates hold as $\tau\to \infty$ for $\beta,\ \gamma$ and $\gamma'$ as in Lemma \ref{timederi}:

\begin{equation}\label{myeq3.61}
\begin{split}
&\lvert\langle E_2\rangle_t\rvert\leq O(\tau^{-1-(\rho-\gamma/2)-\beta/2})(\lVert \phi\rVert^2 -\langle \tilde \theta\pi_{\alpha} g_{\beta,\gamma,\epsilon}(A(\tau),\tau)\pi_{\alpha}\tilde\theta)\rangle_t),\\
&\langle E_3\rangle_t=\langle \theta_{4\epsilon}O(\tau^{\gamma-\beta-1})\theta_{4\epsilon}\rangle_t+O(\tau^{-\infty})\lVert \phi \rVert^2,\\
&\langle E_4\rangle_t\geq \langle \tilde\theta\pi_{\alpha}\tilde g(A(\tau),\tau)\langle x\rangle^{-\rho/2}f_3(H)O(1)f_3(H)\langle x\rangle^{-\rho/2}\tilde g(A(\tau),\tau)\pi_{\alpha}\tilde\theta\rangle_t\\
&\quad\qquad+\langle O(\tau^{-3/2})g_{\beta,(\gamma-\rho-1/2)_+,\epsilon/2}(A(\tau),\tau)\pi_{\alpha}\tilde\theta\rangle_t\\
&\quad\qquad+\langle \tilde\theta\pi_{\alpha}g_{\beta,(\gamma-\rho-1/2)_+,\epsilon/2}(A(\tau),\tau)O(\tau^{-3/2})\rangle_t+O(\tau^{-\infty})\lVert \phi \rVert^2,\\
&\mathrm{with}\ f_3f_2=f_2,\\
&\langle E_5\rangle_t=\langle O(\tau^{-3/2-\rho-\beta/2+\gamma/2})\tilde g(A(\tau),\tau)\pi_{\alpha}\tilde\theta\rangle_t\\
&\quad\qquad+\langle\theta_{2\epsilon}O(\tau^{-3/2-\beta/2+\gamma/2})\tilde g(A(\tau),\tau)\pi_{\alpha}\tilde\theta\rangle_t+O(\tau^{-\infty})\lVert \phi \rVert^2,\\
&\langle E_6\rangle_t=\langle \tilde\theta\pi_{\alpha}\sum_{m=1}^{k}\tilde g^{(m)}(A(\tau),\tau)O(1)\tilde g(A(\tau),\tau)\pi_{\alpha}\tilde\theta\rangle_t,\\
&\langle E_7\rangle_t=\langle O(\tau^{(\gamma-1-\beta)/2-k-1})\tilde g(A(\tau),\tau)\pi_{\alpha}\tilde\theta\rangle_t,\\
&\langle E_8\rangle_t=\langle \tilde\theta\pi_{\alpha}\tilde g(A(\tau),\tau)\sum_{m=1}^{k}O(1)\tilde g^{(m)}(A(\tau),\tau)\pi_{\alpha}\tilde\theta\rangle_t,\\
&\langle E_9\rangle_t=\langle \tilde\theta\pi_{\alpha}\tilde g(A(\tau),\tau)O(\tau^{(\gamma-1-\beta)/2-k-1})\rangle_t,\\
&\langle E_{10}\rangle_t=\langle\tilde\theta\pi_{\alpha}\tilde g(A(\tau),\tau)\sum_{m=0}^{k}(O(\tau^{-3/2-\rho-\beta/2-m+\gamma/2})\\
&\quad\qquad+O(\tau^{-3/2-\beta/2-m+\gamma/2})\theta_{2\epsilon})\rangle_t+O(\tau^{-\infty})\lVert \phi \rVert^2,\\
&\langle E_{11}\rangle_t=\langle \tilde\theta\pi_{\alpha}\tilde g(A(\tau),\tau)O(\tau^{(\gamma-1-\beta)/2-k-1})\rangle_t\\
&\langle E_{12}\rangle_t=\langle \tilde\theta\pi_{\alpha}\sum_{m=2}^{k_1}g_m(A(\tau),\tau)O(\tau^{-(m-\gamma)_+})\sum_{m_1=0}^{[\gamma]+1-m}g_m^{(m_1)}(A(\tau),\tau)\pi_{\alpha}\tilde\theta\rangle_t,\\
&\langle E_{13}\rangle_t=\langle \tilde\theta\pi_{\alpha}\sum_{m=2}^{k_1}g_m(A(\tau),\tau)O(\tau^{-(m-\gamma)_++((\gamma-m)_+-\beta)/2-k_2-1})\rangle_t,\\
&\langle E_{14}\rangle_t=O(\tau^{\gamma-\beta-k_1-1})\lVert \phi \rVert^2.
\end{split}
\end{equation}

As for $E_2$ we notice that
\begin{equation}\label{myeq3.7.0}
[f_1(H)H,\pi_{\alpha}]=\frac{1}{2\pi i}\int_{\mathbb C}\bar \partial_z F_1(z)(H-z)^{-1}[\tilde V_a,\pi_{\alpha}](H-z)^{-1}dz\wedge d\bar z,
\end{equation}
where $F_1$ is the almost analytic extension of $xf_1(x)$ and
\begin{equation}\label{myeq3.7.0.1}
\begin{split}
\tilde V_a(x^a,x_a)\pi_{\alpha}&=\tilde V_a(0,x_a)\pi_{\alpha}+\int_0^1x^a\nabla^a \tilde V_a(rx^a,x_a)dr\pi_{\alpha}\\
&=\pi_{\alpha}\tilde V_a(x^a,x_a)-\pi_{\alpha}\int_0^1x^a\nabla^a \tilde V_a(rx^a,x_a)dr\\
&\quad+\int_0^1x^a\nabla^a \tilde V_a(rx^a,x_a)dr\pi_{\alpha}.
\end{split}
\end{equation}
 
We also notice that
$$[(\nabla_a)_j,\int^1_0(\nabla_a)^{\delta}x^a\nabla^a\tilde V_a(rx^a,x_a)dr]=\int^1_0(\nabla_a)^{\delta'}x^a\nabla^a\tilde V_a(rx^a,x_a)dr.$$
Here, $(\nabla_a)_j$ is the differential with respect to $j$-th component of $x_a$ for some basis in $X_a$, $\delta\in \mathbb N^{\mathrm{dim}X_a}$ and $\delta'=\delta+(0,\dotsm,1,\dotsm,0)$ where $(0,\dotsm,1,\dotsm,0)$ is the vector whose $j$-th component is $1$ and the others are $0$.
Setting $\tilde x=(rx^a,x_a)$,  we have $\lvert \tilde x_{ij}\rvert>C'\lvert \tilde x\rvert>C_1 \lvert x_a\rvert>C_2 \lvert x\rvert$ with some constants $C',C_1,C_2>0$. Thus for any $n\in \mathbb N$ when $\lvert x_a\rvert> C\lvert x^a\rvert$ for sufficiently large $C>0$, it follows that $\lVert\langle p_a\rangle^{2n}\langle x\rangle^{1+\rho}\langle x^a\rangle^{-1-\rho}\nabla^a \tilde V_a(rx^a,x_a)\langle p_a\rangle^{-2n}\rVert_{\mathcal L(\mathcal H)}<C'$ for some $C'>0$ uniformly with respect to $0\leq r\leq 1$.
Thus observing that $\langle x^a\rangle^l\pi_{\alpha}\in \mathcal L(\mathcal H),\ \forall l\in \mathbb R$, it follows that
\begin{equation}\label{myeq3.7.0.3}
\langle p_a\rangle^{2n}\langle x\rangle^{\gamma/2}[f_1(H)H,\pi_{\alpha}]\langle x\rangle^{1+(\rho-\gamma/2)}\langle p_a\rangle^{-2n},[f_1(H)H,\pi_{\alpha}]\langle x\rangle^{-1-\rho}\in\mathcal L(\mathcal H)
\end{equation}
By the interpolation argument we have
\begin{equation}\label{myeq3.7.0.4}
\langle p_a\rangle^{\gamma/2}\langle x\rangle^{\gamma/2}[f_1(H)H,\pi_{\alpha}]\langle x\rangle^{1+(\rho-\gamma/2)}\langle p_a\rangle^{-\gamma/2},\in\mathcal L(\mathcal H)
\end{equation}
Since $\lvert x\rvert>2\epsilon\tau$ on $\mathrm{supp}\ \tilde \theta$ we can see that
\begin{equation}\label{myeq3.7.0.5}
\lVert\tau^{\gamma/2}\langle x\rangle^{-1-\rho}\tilde \theta\rVert\leq C\tau^{-1-(\rho-\gamma/2)},
\end{equation}
and
\begin{equation}\label{myeq3.7.0.6}
\lVert\langle p_a\rangle^{\gamma/2}\langle x\rangle^{-1-(\rho-\gamma/2)}\tilde \theta\langle p_a\rangle^{-\gamma/2}\rVert\leq C\tau^{-1-(\rho-\gamma/2)},
\end{equation}
By the first estimate in \eqref{myeq3.7.0.2} we can see that $\lVert(-g)^{1/2}(A(\tau),\tau)(\langle p_ax_a+x_ap_a\rangle^{-\gamma/2}+\tau^{-\gamma/2})^{-1}\rVert_{\mathcal L(\mathcal H)}\leq C\tau^{-\beta/2}$.
Thus by the second estimate in \eqref{myeq3.7.0.2}, \eqref{myeq3.7.0.3}, \eqref{myeq3.7.0.4}, \eqref{myeq3.7.0.5}, \eqref{myeq3.7.0.6} and $\langle p_a\rangle^{\gamma/2}f_1(H)\in \mathcal L(\mathcal H)$ we obtain
$$\lVert(-g)^{1/2}(A(\tau),\tau)[f_1(H)H,\pi_{\alpha}]\tilde \theta f_1(H)\rVert\leq C\tau^{-1-(\rho-\gamma/2)-\beta/2}.$$
Since $\psi(t)=f_1(H)\psi(t)$ ,
 and Lemma \ref{2cldecay} we obtain
$$\lVert(-g)^{1/2}(A(\tau),\tau) [f_1(H)H,\pi_{\alpha}]\tilde\theta\psi(t)\rVert\leq C\tau^{-1-(\rho-\gamma/2)-\beta/2}\lVert \phi\rVert,$$
for sufficiently large $C>0$.

Therefore by H\"older's inequality and that $ab\leq 2^{-1}(a^2+b^2),\ \forall a,b>0$ we have
\begin{align*}
\lvert \langle E_2\rangle \rvert_t\rvert&\leq \lVert(-g)^{1/2}(A(\tau),\tau) [f_1(H)H,\pi_{\alpha}]\tilde\theta\psi(t)\rVert\lVert(-g)^{1/2}(A(\tau),\tau)\pi_{\alpha}\tilde\theta\psi(t)\rVert\\
&\leq C\tau^{-1-(\rho-\gamma/2)-\beta/2}\lVert \phi\rVert\lVert(-g)^{1/2}(A(\tau),\tau)\pi_{\alpha}\tilde\theta\psi(t)\rVert\\
&\leq 2^{-1}C\tau^{-1-(\rho-\gamma/2)-\beta/2}(\lVert \phi\rVert^2+\lVert(-g)^{1/2}(A(\tau),\tau)\pi_{\alpha}\tilde\theta\psi(t)\rVert^2)\\
&=\leq 2^{-1}C\tau^{-1-(\rho-\gamma/2)-\beta/2}(\lVert \phi\rVert^2-\langle \tilde\theta\pi_{\alpha}g_{\beta,\gamma,\epsilon}(A(\tau),\tau)\pi_{\alpha}\tilde\theta\rangle_t)
\end{align*}
The estimate for $E_3$ follows from $\theta_{2\epsilon}\theta_{\epsilon}=\theta_{\epsilon}$ and $\theta_{\epsilon}\langle x\rangle^r\leq C\tau^r$ for all $r\in \mathbb R$. The estimate for $E_4$ follows easily from the following estimates:
\begin{equation}\label{myeq3.7.1}
\begin{split}
i[H,A(\tau)]&=-2\Delta_a-x_a\nabla_a\tilde V_a,\\
\pi_{\alpha}(-2\Delta_a-x_a\nabla_a\tilde V_a)\pi_{\alpha}&=\pi_{\alpha}(-2\Delta_a-x_a\nabla_a\tilde V_a+2H^a-2\lambda_{\alpha})\pi_{\alpha}\\
&=\pi_{\alpha}(2H-2\lambda_{\alpha}-2\tilde V_a-x_a\nabla_a\tilde V_a)\pi_{\alpha},\\
f(H)(2H-2\lambda_{\alpha}-2\tilde V_a-x_a\nabla_a\tilde V_a)&f(H)\\
\geq f(H)(2&(E-\lambda_{\alpha})-\epsilon_1-2\tilde V_a-x_a\nabla_a\tilde V_a)f(H).
\end{split}
\end{equation}
where $\epsilon_1$ is sufficiently small. The estimate for $E_5$ and $E_{10}$ follows from $f(H)\tilde\theta\pi_{\alpha}(1-f_2(H))=f(H)[\tilde\theta,f_2(H)]\pi_{\alpha}+f(H)\tilde\theta[\pi_{\alpha},f_2(H)]$. The estimates of $E_7$, $E_9$, $E_{11}$, $E_{13}$ and $E_{14}$ follow from that $\lVert (A-z)^{-1}\rVert_{\mathcal L(L^2(X))}\leq \lvert\mathrm{Im}z\rvert$, and for any $M\geq 0$
$$\bar\partial_z\tilde c(z)\leq C_M\langle z\rangle^{(\gamma-1)/2-1-M}\lvert\mathrm{Im}z\rvert^M,$$
and so on. 

We shall prove \eqref{myeq3.6} by showing by induction in $n\in \mathbb N$ the statement $p(n)$ that \eqref{myeq3.6} holds under the further restriction $n-1<\gamma_0\leq n$ (we suppose $p(0)$ holds). 
Assuming $p(n)$ is true, we shall verify $p(n+1)$. First, we shall  prove 
\begin{equation}\label{myeq3.7.2}
\lVert(-g_{0,n,\epsilon}(A(\tau'), \tau'))^{1/2}\pi_{\alpha}\tilde \theta\psi_{\gamma'}(t)\rVert\leq C\lVert \phi \rVert,
\end{equation}
with $2\rho>\gamma'>n$.

\eqref{myeq3.7.2} is obvious for $n=0$. When $n\geq1$, from the assumption  we have
\begin{equation}\label{myeq3.8}
\lVert(-g_{\beta',n,\epsilon}(A(\tau'), \tau'))^{1/2}\pi_{\alpha}\tilde \theta\psi_{\gamma'}(t)\rVert\leq C\lVert \phi \rVert,
\end{equation}
for any $\beta'>0$.
By \eqref{myeq3.0} and \eqref{myeq3.8} we have
\begin{equation}\label{myeq3.8.1}
\lVert(-g_{0,n-1,\epsilon}(A(\tau'), \tau'))^{1/2}\pi_{\alpha}\tilde \theta\psi_{\gamma'}(t)\rVert\leq C\tau^{-(1-\beta')/2}\lVert \phi \rVert,
\end{equation}
 (see Skibsted \cite[Corollary 2.5]{Sk}).

To prove \eqref{myeq3.7.2} we use \eqref{myeq3.6.0}, \eqref{myeq3.61} and \eqref{myeq3.8.1}. The term $\langle E_1\rangle_t$ is positive and therefore negligible. By \eqref{myeq3.61} with $\beta=0$, $2\rho>\gamma'>\gamma=n$, Lemma \ref{Skibstedlem} and \eqref{myeq3.8.1} we have for any $\gamma'-\gamma>\mu>0$
\begin{equation}\label{myeq3.9}
\begin{split}
&\langle E_3\rangle_t= O(t^{-1-\mu})\lVert \phi\rVert^2,\\
&\langle E_4\rangle_t\geq O(t^{-\rho-1+\beta'})\lVert \phi\rVert^2+O(t^{-3/2})\lVert \phi\rVert^2,\\
&\langle E_5\rangle_t,\ \langle E_{10}\rangle_t=O(t^{-2+\beta'/2})\lVert \phi\rVert^2,\\
&\langle E_6\rangle_t,\ \langle E_8\rangle_t=\sum_{m=1}^kO(t^{-1-m+\beta'})\lVert \phi\rVert^2,\\
&\langle E_7\rangle_t,\ \langle E_9\rangle_t,\ \langle E_{11}\rangle_t=O(t^{-3/2+\beta'/2})\lVert \phi\rVert^2,\\
&\langle E_{12}\rangle_t=\sum_{m=2}^{k_1}\sum_{m_1=1}^{[\gamma]+1-m}O(t^{-m-m_1+\beta'})\lVert \phi\rVert^2,\\
&\langle E_{13}\rangle_t=O(t^{-2+\beta'/2})\lVert \phi\rVert^2,\\
&\langle E_{14}\rangle_t=O(t^{-2})\lVert \phi\rVert^2.\\
\end{split}
\end{equation}
Since the right-hand sides of these estimates are integrable for sufficiently small $\beta'$, combining \eqref{myeq3.9} and the estimate for $E_2$ in \eqref{myeq3.61} we have
\begin{align*}
\langle -\tilde \theta \pi_{\alpha}g_{0,n,\epsilon}&(A(\tau'), \tau') \pi_{\alpha}\tilde \theta\rangle_{t'}\\
&\leq C\lVert \phi\rVert^2+C\int_0^{\tau'}\tau^{-1-\rho+n/2} \langle -\tilde \theta \pi_{\alpha}g_{0,n,\epsilon}(A(\tau), \tau) \pi_{\alpha}\tilde \theta\rangle_{t}dt.
\end{align*}
Since $\rho>n/2$, by the Gronwall inequality we obtain \eqref{myeq3.7.2}.

It remains to prove that
\begin{equation}\label{myeq3.9.1}
\lVert(-g_{\beta_0,\gamma_0,\epsilon}(A(\tau'), \tau'))^{1/2}\pi_{\alpha}\tilde \theta\psi_{\gamma'_0}(t)\rVert\leq C\lVert \phi \rVert.
\end{equation}
First, we assume $\rho\geq1$. We use again \eqref{myeq3.6.0} and \eqref{myeq3.61} in conjunction with \eqref{myeq3.7.2} and obtain the following estimates:
\newpage

\begin{equation}\label{myeq3.9.2}
\begin{split}
&\langle E_3\rangle_t= O(t^{-1-\beta_0})\lVert \phi\rVert^2,\\
&\langle E_4\rangle_t\geq O(t^{-\rho-\beta_0})\lVert \phi\rVert^2+O(t^{-3/2})\lVert \phi\rVert^2,\\
&\langle E_5\rangle_t,\ \langle E_{10}\rangle_t=O(t^{-3/2})\lVert \phi\rVert^2,\\
&\langle E_6\rangle_t,\ \langle E_8\rangle_t=\sum_{m=1}^kO(t^{-m-\beta_0})\lVert \phi\rVert^2,\\
&\langle E_7\rangle_t,\ \langle E_9\rangle_t,\ \langle E_{11}\rangle_t=O(t^{-1-\beta_0})\lVert \phi\rVert^2,\\
&\langle E_{12}\rangle_t=\sum_{m=2}^{k_1}\sum_{m_1=1}^{[\gamma]+1-m}O(t^{-m-m_1-\beta_0+1})\lVert \phi\rVert^2,\\
&\langle E_{13}\rangle_t=O(t^{-2+1/2})\lVert \phi\rVert^2,\\
&\langle E_{14}\rangle_t=O(t^{-1-\beta_0})\lVert \phi\rVert^2.\\
\end{split}
\end{equation}
Thus by the Gronwall inequality we obtain \eqref{myeq3.9.1}.

Next we assume $\rho<1$ (note that $E_4$ in \eqref{myeq3.9.2} is not integrable in this case). In this case as in the proof of Skibsted \cite[Corollary 2.6]{Sk}, setting $0<\delta<\rho$ we prove 
\begin{equation}\label{myeq3.9.3}
\lVert(-g_{\beta'_0,\gamma_0,\epsilon}(A(\tau'), \tau'))^{1/2}\pi_{\alpha}\tilde \theta\psi_{\gamma'_0}(t)\rVert\leq C\lVert \phi \rVert,
\end{equation}
with $\beta'_0=\max\{\beta_0,1-\delta\}$. Since $\langle E_4\rangle_t\geq O(t^{-\rho-\beta'_0})\lVert \phi\rVert^2+O(t^{-3/2})\lVert \phi\rVert^2,$ and the right-hand side is integrable by the Gronwall inequality we obtain \eqref{myeq3.9.3}. By iterating we obtain \eqref{myeq3.9.3} with $\beta'_0=\max\{\beta_0,1-n\delta\}$ for any $n\in \mathbb N$ which proves \eqref{myeq3.9.1}.
\end{proof}

\begin{rem}\label{critical term}
The critical term in \eqref{myeq3.6.0} is $E_2$. The terms $E_i,\ 3\leq i\leq14$ are integrable in the proof of \eqref{myeq3.7.2} and \eqref{myeq3.9.1} for $\gamma'>\gamma$ even if $\gamma'\geq 2\rho$. In \eqref{myeq3.61} $\gamma$ appears in the indices of $\tau$ in $E_j,\ j=2-5,7,9-14$. However, for $\langle E_3\rangle_t$ we use Lemma \ref{Skibstedlem} and see that it is integrable also for $\gamma'\geq 2\rho$. As for $\langle E_j\rangle_t,\ j=7,9,11,14$ by the choices of $k$ and $k_1$, we can see that they are integrable. As for $\langle E_5\rangle_t$ and $\langle E_{10}\rangle_t$ we can write
\begin{equation}\label{myeq3.9.4}
\begin{split}
(1-f_2(H))\pi_{\alpha}&=(1-f_2(-\Delta_a+\tilde V_a+\lambda_{\alpha}))\pi_{\alpha}\\
&=\pi_{\alpha}(1-f_2(-\Delta_a+\tilde V_a+\lambda_{\alpha}))-\sum_{j=1}^M\pi_{\alpha}f_2^{(j)}(-\Delta_a+\tilde V+\lambda_{\alpha})\\
&\quad-\frac{1}{2\pi i}\int\bar\partial _zF_2(z)(-\Delta_a+\tilde V_a+\lambda_{\alpha})^{-1}(-1)^M\mathrm{ad}_{-\Delta_a+\tilde V_a}^M([\tilde V_a,\pi_{\alpha}])\\
&\qquad(-\Delta_a+\tilde V_a+\lambda_{\alpha})^{-1-M}dzd\bar z\\
&=\pi_{\alpha}(1-f_2(H))-\sum_{j=1}^M\pi_{\alpha}f_2^{(j)}(H)\\
&\quad-\frac{1}{2\pi i}\int\bar\partial _zF_2(z)(-\Delta_a+\tilde V_a+\lambda_{\alpha})^{-1}(-1)^M\mathrm{ad}_{-\Delta_a+\tilde V_a}^M([\tilde V_a,\pi_{\alpha}])\\
&\qquad(-\Delta_a+\tilde V_a+\lambda_{\alpha})^{-1-M}dzd\bar z,
\end{split}
\end{equation}
where $F_2$ is the almost analytic extention of $f_2$. We have
$$\langle x\rangle^{M+1+\rho}\mathrm{ad}_{-\Delta_a+\tilde V_a}^M([\tilde V_a,\pi_{\alpha}])
(-\Delta_a+\tilde V_a+\lambda_{\alpha})^{-1-M}\in \mathcal L(\mathcal H).$$
Thus, choosing $M=[\gamma]+1$, $\langle E_5\rangle_t$ and $\langle E_{10}\rangle_t$ with $(1-f_2(H))\pi_{\alpha}$ and $\pi_{\alpha}(1-f_2(H))$ being replaced by the third term in the right-hand side of \eqref{myeq3.9.4} is integrable. Since $1-f_2$ and $f_2^{(j)}$, $j\geq 1$ are $0$ on the support of $f_1$, for any $L>0$ there exists $C>0$ such that $\lVert(1-f_2(H))\tilde \theta f_1(H)\rVert\leq C\tau^{-L}$ and  $\lVert f_2^{(j)}(H)\tilde \theta f_1(H)\rVert\leq C\tau^{-L}$. Therefore, noticing $\psi(t)=f_1(H)\psi(t)$ the contribution from the first and the second term in the right-hand side of \eqref{myeq3.9.4} are integrable. As for $\langle E_6\rangle_t$ and $\langle E_8\rangle_t$, assuming $p(n)$ is true we obtain \eqref{myeq3.9} with $\beta=0$, $\gamma'>\gamma=n$ without the restriction $\rho>\gamma'$. We also obtain \eqref{myeq3.9.2} assuming \eqref{myeq3.7.2} holds for $\beta_0>0$ and $n-1<\gamma_0\leq n$.

As for $E_2$ as mentioned in Remark \ref{comparison} the commutator $[\tilde V_a, \pi_{\alpha}]$ in \eqref{myeq3.7.0} cannot be written as
$$[\tilde V_a, \pi_{\alpha}]=O(\langle x\rangle^{-1-\rho})\pi_{\alpha}+O(\langle x\rangle^{-r})$$
where $r>0$ is sufficiently large. For we can write as
$$[\tilde V_a,\pi_{\alpha}]=\int_0^1([x^a,\pi_{\alpha}]\nabla^a\tilde V_a(rx^a,x_a)+x^a[\nabla^a\tilde V_a(rx^a,x_a),\pi_{\alpha}])dr.$$
Since for any $\epsilon$ and  $r,r'$ such that $r+r'=1+\rho+\epsilon$ we have
$$\langle x_a\rangle^{r}\nabla^a\tilde V_a(rx^a,x_a)\langle x_a\rangle^{r'}\notin \mathcal L(\mathcal H),$$
and for $r,r'$ such that $r+r'=2+\rho$,
$$\langle x\rangle^{r}x^a[\nabla^a\tilde V_a(rx^a,x_a),\pi_{\alpha}]\langle x\rangle^{r'}\in \mathcal L(\mathcal H).$$
Since $[x^a,\pi_{\alpha}]$ and $\langle x_a\rangle$ are operators on $L^2(X^a)$ and $L^2(X_a)$ respectively and $L^2(X)=L^2(X^a)\otimes L^2(X_a)$, we can see that $$\langle x_a\rangle^{r}[x^a,\pi_{\alpha}]\nabla^a\tilde V_a(rx^a,x_a)\langle x_a\rangle^{r'}\notin \mathcal L(\mathcal H),$$
for any $\epsilon$ and  $r,r'$ such that $r+r'=1+\rho+\epsilon$. Thus we have 
$$\langle x_a\rangle^{r}[\tilde V_a,\pi_{\alpha}]\langle x_a\rangle^{r'}\notin \mathcal L(\mathcal H),$$
for any $0<\epsilon\leq 1$ and  $r,r'$ such that $r+r'=1+\rho+\epsilon$. 
Therefore, for $\gamma_0\geq 2\rho$ we cannot rewrite $E_2$ in a form which is integrable.
\end{rem}

\begin{proof}[Proof of Theorem \ref{main1}]
Fix $0<E''<E'-\lambda_{\alpha}<E-\lambda_{\alpha}$. Let $M(x_a,\tau):=\left(E''-\frac{x_a^2}{4\tau^2}\right)^{1/2}$, $h\in \mathcal F_{0,1,\epsilon''}$ for $\epsilon''>0$ and $A(\tau)':=h(-\tau M,\tau)$. Since
$$g_{\beta,\gamma,\epsilon}(A(\tau)',\tau)=g_{\beta,\gamma,\epsilon}(-\tau M,\tau),$$
for any $\epsilon>2\epsilon''$, we only need to prove that $\lVert(-g_{\beta_0,\gamma_0,\epsilon}(A(\tau),\tau))^{1/2}\pi_{\alpha}\psi_{\gamma'_0}(t)\rVert\leq C\lVert \phi\rVert^2$ for any $\beta_0>0$, $2\rho>\gamma'_0> \gamma_0\geq 0$, $\epsilon>0$ and $\phi\in \mathcal H$.

Since Lemma \ref{timederi} holds with $A(\tau)$ replaced by $A(\tau)'$, the proof is almost the same as that of Lemma \ref{propagation1} except for the estimate of $E_4$. For the estimate of $E_4$ we have
\begin{align*}
DA(\tau)'&\geq (h^{(1)}(-\tau M,\tau))^{1/2}M^{-1/2}\bigg(\frac{1}{2\tau}\left(\frac{p_ax_a+x_ap_a}{2}-2(E'-\lambda_{\alpha})\tau\right)\\
&\quad  +(E'-\lambda_{\alpha})-E''\bigg)M^{-1/2}(h^{(1)}(-\tau M,\tau))^{1/2}+O(\tau^{-1})\\
&\geq (h^{(1)}(-\tau M,\tau))^{1/2}M^{-1/2}\frac{1}{2}\frac{A(\tau)}{\tau}\chi\left(\frac{A(\tau)}{\tau}<-\epsilon'\right)\\
&\quad \cdot M^{-1/2}(h^{(1)}(-\tau M,\tau))^{1/2}+O(\tau^{-1}),
\end{align*}
where $\epsilon'=E-\lambda_{\alpha}-E''$. Thus for $\gamma'_0>1$ by commuting the factor
$$\left(-\frac{A(\tau)}{\tau}\right)^{1/2}\chi\left(\frac{A(\tau)}{\tau}<-\epsilon'\right),$$
in front of $\pi_{\alpha}\tilde\theta\psi$, $\langle E_4\rangle_t$ is integrable by Lemma \ref{propagation1}.

As for $\gamma'_0\leq 1$ using the estimates \eqref{myeq3.7.0.2}
and $\langle x_a\rangle^{r}\leq C\langle \tau\rangle^{r},\ \forall r>0$ on $\mathrm{supp}\ h^{(1)}(-\tau M,\tau)$, by commuting the factor $\left(-\frac{A(\tau)}{\tau}\right)^{\tilde s}\chi\left(\frac{A(\tau)}{\tau}<-\epsilon'\right)$ in front of $\pi_{\alpha}\tilde\theta\psi$ with $\tilde s<\gamma'_0/2$, $\langle E_4\rangle_t$ is integrable by Lemma \ref{propagation1}.
\end{proof}

\section{Proof of Theorem \ref{main2}, Theorem \ref{main4} and Theorem \ref{main3}}
To prove Theorem \ref{main2} we need the propagation estimate for another operator $\tilde A(\tau)$ which is proved by the estimates of commutators similar to those in the proof of Lemma \ref{propagation1} and Theorem \ref{main1}. Theorem \ref{main2} is proved using covering arguments in the phase space, propagation estimate for $\tilde A(\tau)$ and Theorem \ref{main1}. Theorem \ref{main3} is obtained by the Fourier-Laplace transform and Theorem \ref{main2}.

First we introduce the operator $\tilde A(\tau)$. We set for $\tau=t+t_0,\ t_0>0,\ t\geq0$ and $1>\kappa_0>0$
\begin{equation}\label{myeq4.0}
A_0(\tau)=\frac{\tau}{2}\left\{\frac{x_a}{\tau}\left(p_a-\frac{x_a}{2\tau}\right)+\left(p_a-\frac{x_a}{2\tau}\right)\frac{x_a}{\tau}\right\},
\end{equation}
and
$$I(\tau)=\tau^{2\kappa_0}(\tau^{2\kappa_0}-\Delta_a)^{-1},\ B(\tau)=j(\tau)I(\tau)j(\tau),$$
where
$$j(\tau)=\chi\left(E-\frac{x_a^2}{4\tau^2}<\frac{\lambda_{\alpha}}{3}\right).$$
Setting for $E'>0$
$$\chi(\tau)=\chi\left(\frac{x_a^2}{4\tau^2}-(E'+4)<-1\right),$$
we define
\begin{equation}\label{myeq4.0.1}
\tilde A(\tau)=\chi(\tau)B(\tau)\chi(\tau) A_0(\tau)\chi(\tau)B(\tau)\chi(\tau).
\end{equation}

To prove the propagation estimate of $\tilde A(\tau)$ we need the following lemma of Skibsted \cite{Sk} which holds under Assumption \ref{potentialas}.
\begin{lem}[Skibsted \cite{Sk}]\label{lem4.0}
Let $E>\lambda_{\alpha},\ E\notin \mathcal T,\ \epsilon>0$ be given. Then there exists $E'>0$ such that for any $f\in C_0^{\infty}(\mathbb R)$ supported in a sufficiently small neighborhood of $E$ and any $s\geq l\geq0$,
$$\langle x\rangle^l\chi\left(E'-\frac{x^2}{4t^2}<-\epsilon\right)e^{-itH}f(H)\langle x\rangle^{-s}=O(t^{-s+l})$$
as $t\to \infty$ in $\mathcal L(\mathcal H)$.
\end{lem}

\begin{lem}\label{propagation2}
Suppose the same assumption as in Theorem \ref{main1}. Let $E>\lambda_{\alpha},\ E\notin \mathcal T,\ \epsilon>0$ be given. Then for any $E'>0$ sufficiently large, any $f\in C_0^{\infty}(\mathbb R)$ supported in a sufficiently small neighborhood of $E$, any $s,s',\tilde s\in \mathbb R$ such that $0\leq \tilde s\leq s<\min\{s',\rho\}$, and any $\epsilon>0$,
$$\left(\frac{-\tilde A(\tau)}{\tau}\right)^{\tilde s}\chi\left(\frac{\tilde A(\tau)}{\tau}<-\epsilon\right)\pi_{\alpha}e^{-itH}f(H)\langle x\rangle^{-s'}=O(t^{-s})$$
as $t\to \infty$ in $\mathcal L(\mathcal H)$.
\end{lem}
\begin{proof}
As in the proof of Theorem \ref{main1} we shall prove
\begin{equation}\label{myeq4.0.2}
(\pi_{\alpha}\psi_{\gamma'_0},g_{\beta_0,\gamma_0,\epsilon}(\tilde A(\tau),\tau)\pi_{\alpha}\psi_{\gamma'_0})\leq C\lVert \phi \rVert^2,
\end{equation}
for any $\beta>0$, $\epsilon$ and $2\rho>\gamma'_0>\gamma_0\geq 0$. We rewrite the left hand side of \eqref{myeq4.0.2} as in \eqref{myeq3.1} and the first three terms can be dealt with in the same way as in the proof of Theorem \ref{main1}. Therefore we only need to prove
\begin{equation}\label{myeq4.5}
(\pi_{\alpha}\tilde\theta_{\epsilon}\psi_{\gamma'_0},g_{\beta_0,\gamma_0,\epsilon}(\tilde A(\tau),\tau)\pi_{\alpha}\tilde\theta_{\epsilon}\psi_{\gamma'_0})\leq C\lVert \phi \rVert^2.
\end{equation}

To prove \eqref{myeq4.5} we need several estimates.

As in \cite[Example 4]{Sk}, $i[-\Delta_a+\tilde V_a,\tilde A(\tau)]\langle x^a\rangle^{-1}$ is sum of terms of the form
$$\{h_1Ih_2\dotsm\}O_1(p_a)Ih_mO_1(p_a),$$
where $O_1(p_a)$ is a first order polynomial in components of $p_a$ and $h_j=h_j(x,\tau)$ are smooth and satisfy
$$\partial^{\alpha}_{(x,\tau)}h_j(x,\tau)\langle x^a\rangle^{-\lvert \alpha\rvert}=O(\tau^{-\lvert \alpha\rvert}).$$
The form of $\tilde A(\tau)$ is a finite sum of terms of the form
$$\tau h_1Ih_2O_1(p_a)Ih_3.$$
As for the time derivative of $\tilde A(\tau)$, it is a finite sum of terms of the form
$$\{h_1Ih_2\dotsm\}O_1(p_a)Ih_m.$$

As in \cite[Example 4]{Sk} we set
$$B=(C\tau^{2\kappa_0}+O_1(p_a)\tau^{\kappa_0}+O_2(p_a))(\tau^{2\kappa_0}-\Delta_a)^{-1}.$$
where $O_2(p_a)$ is a second order polynomial in components of $p_a$. We obtain
$$[B,h]=\text{finite sum of terms of the form}\ \tau^{-(1+\kappa_0)}B_1h_1B_2.$$
The form of $\tilde A(\tau)$ is a finite sum of terms of the form
\begin{equation}\label{myeq4.1}
\tau^{1+\kappa_0}h_1B_1h_2B_2h_3,
\end{equation}
$[H,\tilde A(\tau)]$ of terms of the form
\begin{equation}\label{myeq4.2}
\tau^{\kappa_0}h_1B_1h_2\dotsm B_{m-1}h_mO_1(p_a),
\end{equation}
and $d_{\tau}\tilde A(\tau)$ of terms of the form
\begin{equation}\label{myeq4.3}
\tau^{\kappa_0}h_1B_1h_2\dotsm B_{m-1}h_m.
\end{equation}

Moreover among these factors there will always exist at least one of the specific form
\begin{equation}\label{myeq4.4}
B_j=O_1(p_a)\tau^{\kappa_0}(\tau^{2\kappa_0}-\Delta_a)^{-1}.
\end{equation}
For any $n\in \mathbb N$, $n\geq 1$, $ad^n_{\tilde A(\tau)}(H)$ is given by terms of the form \eqref{myeq4.2} and with at least one $B_j$ of the form \eqref{myeq4.4} or of the form $\tau^{2\kappa_0}h_1B_1\dotsm B_{m-1}h_m$ and at least two of the $B_j$ of the form \eqref{myeq4.4}. In particular
\begin{equation}\label{myeq4.4.1}
\mathrm{ad}^n_{\tilde A(\tau)}(H)=\sum h_1B_1\dotsm B_{m-1}h_mO_2(p_a).
\end{equation}

Similarly $\mathrm{ad}_{\tilde A(\tau)}^n(d_{\tau}\tilde A(\tau))$ is given by terms of the form \eqref{myeq4.3} and at least one $B_j$ of the form \eqref{myeq4.4}. Hence
\begin{equation}\label{myeq4.4.2}
\mathrm{ad}_{\tilde A(\tau)}^n(d_{\tau}\tilde A(\tau))=\tau^{\kappa_0}\sum h_1B_1\dotsm B_{m-1}h_m,
\end{equation}
and also
\begin{equation}\label{myeq4.4.3}
\mathrm{ad}_{\tilde A(\tau)}^n(d_{\tau}\tilde A(\tau))=\sum h_1B_1\dotsm B_{m-1}h_mO_1(p_a).
\end{equation}

To estimate the term $E_4$ in Lemma \ref{timederi} with $A(\tau)$ replaced by $\tilde A(\tau)$, we consider $D\tilde A(\tau)$. We have
$$DA_0(\tau)=4\left(p_a-\frac{x_a}{2\tau}\right)^2-2x_a\nabla_a\tilde V_a.$$
Thus
\begin{equation}\label{myeq4.41}
\langle x^a\rangle^{-(1+\rho)/2}\chi(\tau)B(\tau)\chi(\tau)D(A_0(\tau))\chi(\tau)B(\tau)\chi(\tau)\langle x^a\rangle^{-(1+\rho)/2}\geq O(\tau^{-\rho}).
\end{equation}
As for the factor $\chi(\tau)$ we have
\begin{align*}
&D\chi(\tau)=\chi'(\cdot)\frac{x_a}{2\tau^2}p_a+p_a\frac{x_a}{2\tau^2}\chi'(\cdot)-\chi'(\cdot)\frac{x_a^2}{2\tau^3},\\
&\chi'(\cdot)=\frac{d}{dy}\chi(y<-1)_{\vert y=(x_a^2/4\tau^2)-(E'+4)}.
\end{align*}
As in \cite[Example 4]{Sk} the contributions to $f_2(H)D\tilde A(\tau)f_2(H)$ from terms containing such factors is written as
$$\chi_1(\cdot)f_3(H)O(1)f_3(H)\chi_1(\cdot)+O(\tau^{-1})\langle x^a\rangle,$$
where $\chi_1(\cdot)=\chi\left(E'-\frac{x_a^2}{4\tau^2}<-1\right)$ and $f_3\in C_0^{\infty}(\mathbb R)$ satisfies $f_3f_2=f_2$.

We can write
\begin{align*}
f_3(H)\chi_1(\cdot)\tilde g(\tilde A(\tau),\tau)\pi_{\alpha}=&f_3(H)\sum_{m=0}^{k}\tilde g^{(m)}(\tilde A(\tau),\tau)(m!)^{-1}\mathrm{ad}^m_{\tilde A(\tau)}(\chi_1(\cdot))\pi_{\alpha}\\
&+\text{remainder},
\end{align*}
where $k$ is defined by \eqref{myeq3.1.8}.
The right hand side is of the form
\begin{equation}\label{myeq4.42}
O(\tau^{(\gamma-1)/2})\chi\left(E'-\frac{x_a^2}{4\tau^2}<-\frac{1}{2}\right)\pi_{\alpha}+O(\tau^{-1}).
\end{equation}
If $E'$ is large enough and $f\in C_0^{\infty}$ is supported in a sufficiently small neighborhood of $E$, by Lemma \ref{lem4.0} we have
\begin{equation}\label{myeq4.43}
\tau^{(\gamma-1)/2}\chi\left(E'-\frac{x_a^2}{4\tau^2}<-\frac{1}{2}\right)e^{-itH}f(H)\langle x\rangle^{-\gamma'/2}=O(\tau^{(\gamma-1)/2-\gamma'/2}).
\end{equation}

As for the derivative of $j$ we can calculate as in the similar way for $\chi$ using Theorem \ref{main1}. As for the term
$$f_2(H)\chi(\tau) j(\tau)i[\tilde V_a,I]j(\tau)\chi(\tau) A_0(\tau)\chi(\tau)B(\tau)\chi(\tau)f_2(H)$$
we can see easily that
\begin{equation}\label{myeq4.44}
f_2(H)\chi(\tau) j(\tau)i[\tilde V_a,I]j(\tau)\chi(\tau) A_0(\tau)\chi(\tau)B(\tau)\chi(\tau)f_2(H)\langle x^a\rangle^{-1-\rho}=O(\tau^{-\rho}).
\end{equation}

Finally, since we have
$$\frac{d}{d\tau}I(\tau)=-\tau^{2\kappa_0-1}\Delta_a(\tau^{2\kappa_0}-\Delta_a)^{-2},$$
the following estimate holds:
\begin{equation}\label{myeq4.45}
f_2(H)\chi(\tau)j\left(\frac{d}{d\tau}I(\tau)\right)j\chi(\tau)A_0(\tau)\chi(\tau)B(\tau)\chi(\tau)f_2(H)=O(\tau^{-2\kappa_0}).
\end{equation}

Now we shall prove \eqref{myeq4.5}. For any $0\leq\gamma<\gamma'$ and $\beta>0$, Lemma \ref{timederi} holds with $A(\tau)$ replaced by $\tilde A(\tau)$ and the estimates \eqref{myeq3.61} hold except for $\langle E_4\rangle_t$ and $\langle E_{12}\rangle_t$. As for $\langle E_{12}\rangle_t$ the estimates in \eqref{myeq3.61} is replaced by
$$\langle E_{12}\rangle_t=\left\langle \tilde\theta\pi_{\alpha}\sum_{m=2}^{k_1}g_m(A(\tau),\tau)O(\tau^{-(m-\gamma)_++\kappa_0})\sum_{m_1=1}^{[\gamma]+1-m}g_m^{(m_1)}(A(\tau),\tau)\pi_{\alpha}\tilde\theta\right\rangle_t,$$
respectively, where $k_1=\max\{2,[\gamma]+1\}$.
As for $\langle E_4\rangle_t$, by \eqref{myeq4.41}-\eqref{myeq4.45} we have
\begin{align*}
\langle E_{4}\rangle_t\geq& \langle \tilde\theta\pi_{\alpha}O(\tau^{-2})\pi_{\alpha}\tilde\theta\rangle_t+\langle \tilde\theta\pi_{\alpha}\tilde gO(\tau^{-\rho})\tilde g\pi_{\alpha}\tilde\theta\rangle_t\\
&+\langle \tilde\theta\pi_{\alpha}\tilde gO(\tau^{-2\kappa_0})\tilde g\pi_{\alpha}\tilde\theta\rangle_t+O(\tau^{(\gamma-1)-\gamma'})\lVert \phi\rVert^2.
\end{align*}
Thus the proof of \eqref{myeq4.5} is almost the same as that of \eqref{myeq3.6} in the proof of Lemma \ref{propagation1} except that we set $\delta<\min\{2\kappa_0,\rho\}$ and $\beta'_0=\max\{\beta_0,1-n\delta\}$ as in the last part of the proof of Lemma \ref{propagation1} regardless of whether $\rho<1$ or not.
\end{proof}

To prove Theorem \ref{main2} we also need the following lemmas.
\begin{lem}\label{lem4.1}
Let $\Omega\subset (\lambda_{\alpha},+\infty)$ be a compact interval, and $f\in C_0^{\infty}(\dot \Omega)$ where $\dot \Omega$ is the interior of $\Omega$. Then for any $0<s<\min\{s',\rho\}$ we have
$$\langle x_a\rangle^{-s}\pi_{\alpha}e^{-itH}f(H)\langle x\rangle^{-s'}=O(t^{-s}),$$
as $t\to \infty$ in $\mathcal L(\mathcal H)$.
\end{lem}
\begin{proof}
For $E\in \Omega$, $\epsilon>0$ and $f_1\in C_0^{\infty}(\mathbb R)$ supported in a neighborhood of $E$,
\begin{align*}
\langle x_a\rangle^{-s}\pi_{\alpha}e^{-itH}&f_1(H)\langle x\rangle^{-s'}\\
&=\chi\left(\frac{x_a^2}{4t^2}-(E-\lambda_{\alpha})<-\epsilon\right)\langle x_a\rangle^{-s}\pi_{\alpha}e^{-itH}f_1(H)\langle x\rangle^{-s'}\\
&\quad+\left(1-\chi\left(\frac{x_a^2}{4t^2}-(E-\lambda_{\alpha})<-\epsilon\right)\right)\langle x_a\rangle^{-s}\pi_{\alpha}e^{-itH}f_1(H)\langle x\rangle^{-s'}\\
&=\chi\left(\frac{x_a^2}{4t^2}-(E-\lambda_{\alpha})<-\epsilon\right)\langle x_a\rangle^{-s}\pi_{\alpha}e^{-itH}f_1(H)\langle x\rangle^{-s'}\\
&\quad+O(t^{-s}).
\end{align*}
If $f_1$ is supported in sufficiently small neighborhood of $E$, by Theorem \ref{main1} we have
$$\chi\left(\frac{x_a^2}{4t^2}-(E-\lambda_{\alpha})<-\epsilon\right)\langle x_a\rangle^{-s}\pi_{\alpha}e^{-itH}f_1(H)\langle x\rangle^{-s'}=O(t^{-s}).$$
By a covering argument the lemma follows.
\end{proof}

\begin{lem}\label{lem4.2}
Let $\Omega\subset (\lambda_{\alpha},+\infty)$ be a compact interval, $f\in C_0^{\infty}(\dot \Omega)$ and $p(x_a,\xi_a)\in S^m_l,\ m,l\in \mathbb R$ satisfy  $p(x_a,\xi_a)=0$ if $\xi_a^2\in\Omega-\lambda_{\alpha}$. Then for any $s>0$
$$p(x_a,D_a)\pi_{\alpha}f(H)\langle x\rangle^s\in\mathcal L(\mathcal H)$$
\end{lem}
\begin{proof}
For a closed curve $\Gamma$ in $\mathbb C$ around $\mathrm{supp}\ f$ and intersecting $\mathbb R$ in $\dot\Omega$ we have
\begin{align*}
p(x_a,D_a)\pi_{\alpha}f(H)&=p(x_a,D_a)\frac{1}{2\pi i}\oint_{\Gamma-\lambda_{\alpha}}(-\Delta_a-z)^{-1}dz\pi_{\alpha}f(H)\\
&\quad-p(x_a,D_a)\pi_{\alpha}\frac{1}{2\pi i}\oint_{\Gamma}(H-z)^{-1}dzf(H)\\
&=\frac{1}{2\pi i}\oint_{\Gamma}p(x_a,D_a)((-\Delta_a-z+\lambda_{\alpha})^{-1}\pi_{\alpha}\\
&\quad -\pi_{\alpha}(H-z)^{-1})dzf(H)\\
&=\frac{1}{2\pi i}\oint_{\Gamma}p(x_a,D_a)(-\Delta_a-z+\lambda_{\alpha})^{-1}\pi_{\alpha}\tilde V_a(H-z)^{-1}dzf(H).
\end{align*}
Thus there is a $x$-decay by a factor $\langle x\rangle^{-\rho}$.  We have
\begin{align*}
p(x_a,D_a)&(-\Delta_a-z+\lambda_{\alpha})^{-1}\pi_{\alpha}\tilde V_a(H-z)^{-1}f(H)\\
&=(-\Delta_a-z+\lambda_{\alpha})^{-1}\pi_{\alpha}R\tilde p(x_a,D_a)f(H)(H-z)^{-1}+\tilde R(H-z)^{-1}f(H),
\end{align*}
where $R,\tilde R \in \mathcal L(\mathcal H)$,$\tilde p$ satisfies the same condition as $p$, $\pi_{\alpha}R\langle x\rangle^{\rho}\in \mathcal L(\mathcal H)$ and $\tilde R\langle x\rangle^n\in L(\mathcal H),\ \forall n\in \mathbb N$.  Repeating the same arguments as above for $\tilde p(x_a,D_a)f(H)$ we can see that $p(x_a,D_a)\pi_{\alpha}f(H)\langle x\rangle^s$ is bounded.
\end{proof}

We need the following lemma of Skibsted \cite{Sk} which holds under Assumption \ref{potentialas}.
\begin{lem}[Skibsted \cite{Sk}]\label{lem4.3}
With $f$ as in Lemma \ref{lem4.2} for any $0<s<s'$,
$$\langle x\rangle^{-s}e^{-itH}f(H)\langle x\rangle^{-s'}=O(t^{-s}),$$
as $t\to \infty$ in $\mathcal L(\mathcal H)$.
\end{lem}

The next lemma is an easy consequence of Lemma \ref{lem4.2} and \ref{lem4.3}.
\begin{lem}\label{lem4.4}
With the condition of Lemma \ref{lem4.2} for any $s,s'$ such that $\min\{s',\rho\}>s>0$ and $\rho-l>s$,
$$p(x_a,D_a)\pi_{\alpha}e^{-itH}f(H)\langle x\rangle^{-s'}=O(t^{-s}),$$
as $t\to +\infty$ in $\mathcal L(\mathcal H)$.
\end{lem}

\begin{proof}[Proof of Theorem \ref{main2}]
We may suppose $s'<\rho$. Let $E>0$ and $\xi_0\in X_a$ with $\lvert \xi_0\rvert^2=E-\lambda_{\alpha}$ be given. By Lemma \ref{lem4.1}, Lemma \ref{lem4.4} and covering arguments it suffices to find neighborhoods $N_E$ of $E$ and $N_{\xi_0}$ of $\xi_0$, respectively such that the estimate holds for any $f\in C_0^{\infty}(\mathbb R^+)$ and $p(x_a,D_a)\in S^0_0$ with the properties: $\mathrm{supp}\ f\in N_E,\ \mathrm{supp}\ p\subset\{(x_a,\xi_a)\vert x_a\cdot\xi_a<(1-\epsilon_1)\lvert x_a\rvert \lvert \xi_a\rvert\}$ and $p(x_a,\xi_a)=0$ for $\xi_a\notin N_{\xi_0}$.

We write
\begin{align*}
\langle x_a\rangle^lp(x_a,D_a)&\pi_{\alpha}e^{-itH}f(H)\langle x\rangle^{-s'}\\
=&\langle x_a\rangle^l\chi\left(E'-\frac{x_a^2}{4t^2}<-2\epsilon\right)p(x_a,D_a)\pi_{\alpha}e^{-itH}f(H)\langle x\rangle^{-s'}\\
&+\langle x_a\rangle^l\left(1-\chi\left(E'-\frac{x_a^2}{4t^2}<-2\epsilon\right)\right)p(x_a,D_a)\pi_{\alpha}e^{-itH}f(H)\langle x\rangle^{-s'}.
\end{align*}
with $E'$ and $\epsilon$ as in Lemma \ref{lem4.0}.
Since $E'-\frac{x^2}{4t^2}\leq E'-\frac{x_a^2}{4t^2}<-2\epsilon$ holds on $\mathrm{supp}\ \chi\left(E'-\frac{x_a^2}{4t^2}<-2\epsilon\right)$, we have 
\begin{equation}\label{myeq4.8.1}
\chi\left(E'-\frac{x_a^2}{4t^2}<-\epsilon\right)=\chi\left(E'-\frac{x_a^2}{4t^2}<-2\epsilon\right)\chi\left(E'-\frac{x^2}{4t^2}<-\epsilon\right).
\end{equation}
By the pseudodifferential calculus, \eqref{myeq4.8.1} and Lemma \ref{lem4.0} the first term is estimated as $O(t^{-s+l})$. As for the second term we have
\begin{align*}
\langle x_a\rangle^l\left(1-\chi\left(E'-\frac{x^2}{4t^2}<-2\epsilon\right)\right)&p(x_a,D_a)\pi_{\alpha}e^{-itH}f(H)\langle x\rangle^{-s'}\\
&=O(t^l)p(x_a,D_a)\pi_{\alpha}e^{-itH}f(H)\langle x\rangle^{-s'}.
\end{align*}
Therefore, we only need to prove the estimate with $l=0$.

Let $\chi_1(x_a)\in C^{\infty}(X_a)$ be homogeneous of degree zero outside the unitsphere in $X_a$ and satisfy
\begin{equation}\label{myeq4.9}
\mathrm{supp}\ \chi_1(\cdot)\subset\left\{x_a\vert x_a\cdot\xi_0<\left(1-\frac{\epsilon_1}{2}\right)\lvert x_a\rvert \lvert\xi_0\rvert\right\}.
\end{equation}
Corresponding to $E$ given above and $\epsilon=\frac{1}{3}$ in Lemma \ref{lem4.0} we can find a neighborhood $N^1_E$ of $E$ and $E'\geq E$, such that the estimates of the Lemma hold for $E'$ and for any $f$ with $\mathrm{supp}\ f\subset N^1_E$. With this $E'$ we define $\tilde A(\tau)$ by \eqref{myeq4.0.1}. 

We choose $0<\epsilon<\min\left\{\frac{(-\lambda_{\alpha})}{18},\frac{\epsilon_1(E-\lambda_{\alpha})}{2}\right\}$ so small that 
\begin{equation}\label{myeq4.10}
(E-\lambda_{\alpha}-6\epsilon)^{1/2}\left((E-\lambda_{\alpha}-6\epsilon)^{1/2}\left(1-\frac{\epsilon_1}{4}\right)-\left(1-\frac{\epsilon_1}{2}\right)(E-\lambda_{\alpha})^{1/2}\right)>\frac{\epsilon_1}{8}(E-\lambda_{\alpha}),
\end{equation}
and with $d=\frac{\epsilon_1}{4}(E-\lambda_{\alpha})\left(1-\frac{\epsilon_1}{4}\right)$
\begin{equation}\label{myeq4.11}
2(E'+2)^{1/2}3\epsilon<\frac{d}{6}.
\end{equation}

Corresponding to $E$ given above and $\epsilon$ above in Theorem \ref{main1} we can find a neighborhood $N^2_E$ of $E$, such that the estimate of the theorem holds for any $f$ with $\mathrm{supp}\ f\subset N^2_E$. Put $N_E=N^1_E\cap N^2_E$. Clearly Lemma \ref{propagation2} holds for any $f$ with $\mathrm{supp}\ f\subset N_E$.

Let $\chi_2(\xi_a)$ be a smooth function such that
$$\mathrm{supp}\ \chi_2(\cdot)\subset B_{\epsilon}(\xi_0)=\{\xi_a\vert \lvert\xi_a-\xi_0\rvert\leq \epsilon\},$$
with $\epsilon$ above.
By Lemma \ref{lem4.1} and the calculus of pseudodifferential operators it is sufficient to prove for $\phi\in \mathcal H$ and $\mathrm{supp}\ f\subset N_E$ with $\psi(t)=e^{-itH}f(H)\langle x\rangle^{-s'}\phi,\ \chi_2=\chi_2(D_a)$ and $\psi_1(t)=\chi_1\chi_2\pi_{\alpha}\psi(t)$ that
\begin{equation}\label{myeq4.11.1}
\lVert \psi_1(t)\rVert = O(t^{-s})\lVert \phi\rVert.
\end{equation}

Set
$$B(t)=\left(I-\chi\left(\frac{x_a^2}{4t^2}-(E-\lambda_{\alpha})<-3\epsilon\right)\right)\left(I-\chi\left(E'-\frac{x_a^2}{4t^2}<-1\right)\right).$$

As in Skibsted \cite{Sk} we have
$$\lVert \psi_1(t)\rVert \leq 2\lVert B(t)\psi_1(t)\rVert^2+Ct^{-2s}\lVert \phi\rVert.$$
Due to \eqref{myeq4.9} and \eqref{myeq4.10} we have for $x_a\in \mathrm{supp}\ B(t)$ and $\tau=t+t_0,\ t_0>0,\ t>t_0\left(\frac{4}{\epsilon_1}-1\right)$,
\begin{equation*}
\frac{x_a}{\tau}\left(\frac{x_a}{2\tau}-\xi_0\right)\geq\left(1-\frac{\epsilon_1}{4}\right)\frac{\epsilon_1}{4}(E-\lambda_{\alpha})=d.
\end{equation*}
Thus we have
$$\lVert \psi_1\rVert^2\leq \frac{2}{d}\left(B(t)\psi_1(t),\frac{x_a}{\tau}\left(\frac{x_a}{2\tau}-\xi_0\right)B(t)\psi_1(t)\right)+Ct^{-2s}\lVert \phi\rVert^2.$$

Pick $\chi_3(\xi_a)\in C^{\infty}(X_a')$ such that $\mathrm{supp}\ \chi_3\subset B_{3\epsilon}(\xi_0)$ and $\chi_3(\xi_a)=1$ on $B_{2\epsilon}(\xi_0)$. Then
\begin{equation}\label{myeq4.12}
(I-\chi_3(D_a))B(t)\chi_1\chi_2=O(t^{-2s})
\end{equation}
and therefore,
\begin{equation}\label{myeq4.13}
\begin{split}
\lVert \psi_1\rVert^2\leq& -\frac{2}{d}\left(B(t)\psi_1(t),\frac{A_0(\tau)}{\tau}B(t)\psi_1(\tau)\right)\\
&+\frac{2}{d}\left \lVert B(t)\frac{x_a}{\tau}(p_a-\xi_0)\chi_3(D_a)B(t)\right \rVert\lVert \psi_1\rVert^2+Ct^{-2s}\lVert \phi\rVert^2
\end{split}
\end{equation}
By \eqref{myeq4.11} we have
\begin{equation}\label{myeq4.14}
\left \lVert B(t)\frac{x_a}{\tau}(p_a-\xi_0)\chi_3(D_a)B(t)\right \rVert\leq \frac{d}{6}.
\end{equation}

As for the first term of the right hand side of \eqref{myeq4.13}, we notice that
\begin{equation}\label{myeq4.15}
\begin{split}
-\left(B(t)\psi_1(t),\frac{A_0(\tau)}{\tau}B(t)\psi_1(\tau)\right)\leq& -\left(B(t)\psi_1(t),\frac{\tilde A(\tau)}{\tau}B(t)\psi_1(\tau)\right)\\
&+Ct^{-2s}\lVert \phi\rVert^2+C_1t^{-2\kappa_0}\lVert \psi_1\rVert^2
\end{split}
\end{equation}

We choose $t'>0$ such that
\begin{equation}\label{myeq4.16}
C_1t^{-2\kappa_0}<\frac{d}{6}
\end{equation}

We obtain by \eqref{myeq4.12}-\eqref{myeq4.16} and transposition that for
$$t>\max\left\{t_0\left(\frac{4}{\epsilon_1}-1\right),t'\right\}$$
with $\epsilon'=d/12$
\begin{equation}\label{myeq4.16.1}
\frac{1}{6}\lVert \psi_1\rVert^2\leq-\frac{2}{d}\left(B(t)\psi_1(t),\frac{\tilde A(\tau)}{\tau}\chi^2\left(\frac{\tilde A(\tau)}{\tau}<-\epsilon'\right)B(t)\psi_1(\tau)\right)+Ct^{-2s}\lVert \phi\rVert^2.
\end{equation}
For $s'>1/2$, commuting the factor $\left(-\frac{\tilde A(\tau)}{\tau}\right)^{1/2}\chi\left(\frac{\tilde A(\tau)}{\tau}<-\epsilon'\right)$ in front of $\pi_{\alpha}\psi(t)$, by Lemma \ref{propagation2} we have \eqref{myeq4.11.1}.

When  $s'\leq 1/2$, using that $\lVert\langle \tilde A(\tau)\rangle^r\langle p_a\rangle^{-r}\rVert_{\mathcal L(\mathcal H)}=O(\tau^r)$ for any $r>0$ and $\langle p_a\rangle^r\chi_2(D_a)\in \mathcal L(\mathcal H)$,  commuting the factor $\left(-\frac{\tilde A(\tau)}{\tau}\right)^{\tilde s}\chi\left(\frac{\tilde A(\tau)}{\tau}<-\epsilon'\right)$ with $\tilde s<s'$ in front of $\pi_{\alpha}\psi(t)$ we obtain \eqref{myeq4.11.1} which completes the proof of Theorem \ref{main2}.
\end{proof}

\begin{proof}[Idea of the proof of Theorem \ref{main4}]
The proof is very similar to that of Theorem \ref{main2} (see also Skibsted \cite[Theorem 4.5]{Sk}). By Lemma \ref{lem4.0} it can be assumed that $l=0$. We introduce cutoff functions $\chi_1(x_a)$ and $\chi_2(\xi_a)$ with similar properties except for \eqref{myeq4.9} and the operator $\tilde A(\tau)$ in \eqref{myeq4.0.1}. Then since we have
$$\frac{x_a}{\tau}\left(\frac{x_a}{\tau}-\xi_0\right)\geq 2\sqrt{E-\lambda_{\alpha}+\epsilon}(\sqrt{E-\lambda_{\alpha}+\epsilon}-\sqrt{E-\lambda_{\alpha}})>0,$$
we derive an estimate of the type \eqref{myeq4.16.1} but with 
$$B(t)=\left(I-\chi\left(E'-\frac{x_a^2}{4t^2}<-1\right)\right),$$
and 
$$\psi_1(t)=\chi\left(E-\lambda_{\alpha}-\frac{x_a^2}{4t^2}<-\epsilon\right)\chi_1\chi_2\pi_{\alpha}\psi(t).$$
and apply it with Lemma \ref{propagation2}.
\end{proof}

\begin{proof}[Proof of Theorem \ref{main3}]
We may suppose $s'<\rho$. Let $f$ be a smooth function such that $\mathrm{supp}\ f\subset \mathbb R^+$ and $f=1$ on some compact interval $I_1\subset \mathbb R^+$ such that $I\subset \dot I_1$.
We can easily see that there exists a constant $C>0$ such that
$$\sup_{\substack{\lambda\in I\\ \mu>0}}\lVert\langle x_a\rangle^sp(x_a,D_a)\pi_{\alpha}(R(\lambda+i\mu))^m(1-f(H))\langle x\rangle^{-s'}\rVert<C,$$
and the limit
$$s-\lim_{\substack{\mu\in \mathbb R^+\\\mu\to +0}}\langle x_a\rangle^s p(x_a,D_a)\pi_{\alpha}(R(\lambda+i\mu))^m(1-f(H))\langle x\rangle^{-s'},$$
exists.

We have for $z\in \mathbb C,\ \mathrm{Im}\ z>0$,
\begin{equation}\label{myeq4.18}
\begin{split}
&\langle x_a\rangle^sp(x_a,D_a)\pi_{\alpha}(R(z))^mf(H)\langle x\rangle^{-s'}\\
&\quad=i^m\int_0^{\infty}\int_{t_m}^{\infty}\dotsm\int_{t_2}^{\infty}\langle x_a\rangle^sp(x_a,D_a)\pi_{\alpha}e^{-it_1(H-z)}f(H)\langle x\rangle^{-s'}dt_1\dotsm dt_m,
\end{split}
\end{equation}
where $\int_{t_m}^{\infty}\dotsm\int_{t_2}^{\infty}$ and $\dotsm dt_m$ are removed when $m=1$.
By Theorem \ref{main3} the right-hand side of \eqref{myeq4.18} is bounded uniformly with respect to $z\in \mathbb C$ such that $\mathrm{Re}\ z\in I$ and $\mathrm{Im} z>0$ which proves the first part of Theorem \ref{main3}.

As for the second part, for $z_1,\ z_2\in \mathbb C$ such that $\mathrm{Re}\ z_j\in I$ and $\mathrm{Im} z_j>0$, $j=1,2$ we have
\begin{equation}\label{myeq4.19}
\begin{split}
&\langle x_a\rangle^sp(x_a,D_a)\pi_{\alpha}(R(z_1))^mf(H)\langle x\rangle^{-s'}\\
&\qquad-\langle x_a\rangle^sp(x_a,D_a)\pi_{\alpha}(R(z_2))^mf(H)\langle x\rangle^{-s'}\\
&\quad =i^m\int_0^{\infty}\int_{t_m}^{\infty}\dotsm\int_{t_2}^{\infty}\langle x_a\rangle^sp(x_a,D_a)\pi_{\alpha}e^{-it_1(H-z_1)}f(H)\langle x\rangle^{-s'}dt_1\dotsm dt_m\\
&\qquad-i^m\int_0^{\infty}\int_{t_m}^{\infty}\dotsm\int_{t_2}^{\infty}\langle x_a\rangle^sp(x_a,D_a)\pi_{\alpha}e^{-it_1(H-z_2)}f(H)\langle x\rangle^{-s'}dt_1\dotsm dt_m\\
&\quad =i^m\int_0^{\infty}\int_{t_m}^{\infty}\dotsm\int_{t_2}^{\infty}\langle x_a\rangle^sp(x_a,D_a)\pi_{\alpha}\\
&\qquad \cdot\int_0^1it_1(z_1-z_2)e^{-it_1(H-z_2+r(z_2-z_1))}drf(H)\langle x\rangle^{-s'}dt_1\dotsm dt_m.
\end{split}
\end{equation}
Setting $\mu_j=\mathrm{Im}z_j,\ j=1,2$ and $0<2s_1<s'-s-m$ and assuming $\mathrm{Re} z_1=\mathrm{Re} z_2$ and $\mu_1>\mu_2>0$ the norm of the right-hand side of \eqref{myeq4.19} is estimated as
\begin{align*}
C(\mu_1-\mu_2)&\int_0^{\infty}\int_{t_m}^{\infty}\dotsm\int_{t_2}^{\infty}\int_0^1t_1^{-m-2s_1+1}e^{-(\mu_2+r(\mu_1-\mu_2))t_1}drdt_1\dotsm dt_m\\
&\leq C(\mu_1-\mu_2)\int_0^{\infty}\int_{t_m}^{\infty}\dotsm\int_{t_3}^{\infty}t_2^{-m-s_1+1}dt_2\dotsm dt_m\\
&\quad\cdot\int_0^1\int_{0}^{\infty}t_1^{-s_1}e^{-(\mu_2+r(\mu_1-\mu_2))t_1}dt_1dr\\
&\leq C(\mu_1-\mu_2)\int_0^1(\mu_2+r(\mu_1-\mu_2))^{-1+s_1}dr\leq C(\mu_1-\mu_2)^{s_1},
\end{align*}
where $\int_0^{\infty}\int_{t_m}^{\infty}\dotsm\int_{t_3}^{\infty}t_2^{-m-s_1}dt_2\dotsm dt_m$ is removed when $m=1$ and $\int_{t_m}^{\infty}\dotsm\int_{t_3}^{\infty}$ and $\dotsm dt_m$ are removed when $m=2$.
This proves the existence of the limit. The continuity can be seen in a similar way.
\end{proof}

\section{Proof of Theorem \ref{projection}, Theorem \ref{ais1}, Theorem \ref{three} and Theorem \ref{ais2}}\label{sec5}
\begin{proof}[Idea of the proof of Theorem \ref{projection}]
The construction of $\Pi$ is similar to the construction of the projection onto the almost invariant subspaces in the Born--Oppenheimer approximation for atoms and molecules (see Ashida \cite{As} and Martinez and Sordoni \cite{MS1,MS2}).

In the Born--Oppenheimer approximation we considered the smallness of operators with respect to the powers of the semiclassical parameter $h$ which is the square root of the ratio of electronic and nuclear mass. In propagation estimates we measure the smallness of operators with respect to the decay in intercluster coordinates $x_a$. In the construction of $\Pi$ we need pseudodifferential operators with operator valued symbols (see Ashida \cite[Appendix]{As} and Martinez and Sordoni \cite{MS2}). Let $\mathcal H_1$ and $\mathcal H_2$ be Hilbert spaces. A function $a(x_a,\xi_a)\in C^{\infty}(X_a \times X_a';\mathcal L(\mathcal H_1,\mathcal H_2))$ is said to be in $S^m_l(\mathcal L(\mathcal H_1,\mathcal H_2))$ if for any $\alpha,\beta\in \mathbb N^{\mathrm{dim}X_a}$ one has
$$\lVert \partial_{x_a}^{\alpha}\partial_{\xi_a}^{\beta}p(x_a,\xi_a)\rVert_{\mathcal L(\mathcal H_1,\mathcal H_2)}\leq C_{\alpha,\beta}\langle x_a\rangle^{l-\lvert \alpha\rvert}\langle \xi_a\rangle^m,\ \forall (x_a,\xi_a)\in X_a\times X_a',$$
For $a(x_a,\xi_a)\in S^m_l(\mathcal L(\mathcal H_1,\mathcal H_2))$ and $u\in \mathcal S(X_a; \mathcal H_1)$ we set
$$(a(x_a,D_a)u)(x_a):=(2\pi)^{-\nu \mathrm{dim}X_a}\int\int e^{i(x_a-y_a)\xi_a}a(x_a,\xi_a)u(y_a)dy_ad\xi_a.$$
Since $a(x_a,D_a)$ is continuous $\mathcal S(X_a;\mathcal H_1)\to \mathcal S(X_a;\mathcal H_2)$, we can extend it uniquely to a linear continuous operator $\mathcal S'(X_a;\mathcal H_1)\to \mathcal S'(X_a;\mathcal H_2)$.

We denote by $h(x_a,\xi_a):=\xi_a^2+\tilde V_a+H^a$ the operator valued symbol of $H$@and by $h_0(x_a,\xi_a):=\xi_a^2+H^a$ its principal symbol with respect to the decay in intercluster coordinates $x_a$ (note that $\tilde V_a\in S_{-\rho}^0(\mathcal L(e^{-c\lvert x^a\rvert}L^{2}(X^a),L^{2}(X^a))$ for any $c>0$). Let $\gamma$ be a loop enclosing $\sigma$ and satisfying $\mathrm{dist}(\gamma,\sigma)>0$. Set $\Omega:=\{(x_a,\xi_a,z)\in X_a\times X_a'\times \mathbb C; z-\xi_a^2\in \gamma\}$. For$(x_a,\xi_a,z)\in \Omega$, $h_0(x_a,\xi_a)-z$ is invertible and $q_0(x_a,\xi_a;z):=(h_0(x_a,\xi_a)-z)^{-1}$ is smooth and bounded. We define a symbol $r(x_a,\xi_a;z)$ as in Ashida \cite{As} by 
$$r(x_a,\xi_a;z):=1-(h(x_a\xi_a)-z)\# q_0(x_a,\xi_a;z).$$
Then we can easily see that $r\in S_{-1-\rho}^0(\mathcal L(e^{-c\lvert x^a\rvert}L^{2}(X^a),L^{2}(X^a))$ for any $c>0$. We define $q(x_a,\xi_a;z)$ by
$$q(x_a,\xi_a;z):=q_0+q_0\#\sum_{j\geq 1}r^{\# j}.$$
Then $q(x_a,\xi_a;z)$ can be rewritten as
$$q(x_a,\xi_a;z)=\sum_{j \geq 1} q_j(x_a,\xi_a;z),$$
where $q_j(x_a,\xi_a;z)\in S_{-j-\rho}^j(\mathcal L(e^{-c\lvert x^a\rvert}L^{2}(X^a),L^{2}(X^a))$ for any $c>0$. Moreover, $q_j(x_a,\xi_a;z)$ is given by the sum of terms of the following form
$$Cq_0\prod_{i=1}^m((\partial_{x_a}^{\alpha_i}\partial_{\xi_a}^{\beta_i}s_i)q_0),$$
where $C$ is a constant $1\leq m\leq 2j$, $\lvert \sum_{i=1}^m\alpha_i\rvert=\lvert \sum_{i=1}^m\beta_i\rvert\leq j$ and $s_i$ is one of the following
$$\xi_a^{\alpha}, \lvert \alpha\rvert=1,\ \partial_{x_a}^{\beta} \tilde V_a, \lvert \beta\rvert=1.$$
Let us set $\Gamma(\xi_a):=\{z\in \mathbb C;z-\xi_a^2\in \gamma\}$, and define
$$\hat \pi_j(x_a,\xi_a):=\frac{i}{2\pi}\oint_{\Gamma(\xi_a)}q_j(x_a,\xi_a;z)dz.$$
Then we have $\hat \pi_0=\pi$. Moreover, as in Ashida \cite[Lemma 3.2]{As} we can insert $\pi$ in the terms in $\hat \pi_j$ and prove that $q_j(x_a,\xi_a;z)\in S_{-j-\rho}^j(\mathcal L(L^{2}(X^a))$ and
$$\Phi(H)\hat \pi_j(x_a,D_a),\Phi(H)\hat \pi_j(x_a,D_a)H\in\mathcal L(\mathcal H),$$

Let $\delta\in C_0^{\infty}(\mathbb R)$ be a function such that $0\leq \delta\leq 1$, $\mathrm{supp}\, \delta\subset [-2,2]$ and $\delta=1$ on $[-1,1]$. Then it is easy to see that there exists a decreasing sequence of positive numbers $(\epsilon_j)_{j\in \mathbb N}$ converging to zero  such that for any $j\in \mathbb N$ 
\begin{align*}
\lVert&\langle D\rangle^u\langle x\rangle^r(1-\delta(\epsilon_j\lvert x_a\rvert))\hat \pi_j(x_a,D_a)\Phi(H)\langle x \rangle^{r'}\langle D\rangle^{u'}\rVert_{\mathcal L(\mathcal H)}\\
&+\lVert\langle D\rangle^u\langle x\rangle^r(1-\delta(\epsilon_j\lvert x_a\rvert))\hat \Phi(H)\pi_j(x_a,D_a)\langle x \rangle^{r'}\langle D\rangle^{u'}\rVert_{\mathcal L(\mathcal H)}\\
&+\lVert\langle D\rangle^u\langle x\rangle^r(1-\delta(\epsilon_j\lvert x_a\rvert))H\hat \pi_j(x_a,D_a)\Phi(H)\langle x \rangle^{r'}\langle D\rangle^{u'}\rVert_{\mathcal L(\mathcal H)}\\
&+\lVert\langle D\rangle^u\langle x\rangle^r(1-\delta(\epsilon_j\lvert x_a\rvert))\hat \Phi(H)\pi_j(x_a,D_a)H\langle x \rangle^{r'}\langle D\rangle^{u'}\rVert_{\mathcal L(\mathcal H)}\leq 2^{-j-2},
\end{align*}
for any $r,r',u,u'\in \mathbb R$ such that $\lvert r\rvert,\lvert r'\rvert,\lvert u\rvert,\lvert u'\rvert\leq j$ and $r+r'\leq j$.
Let us define
\begin{align*}
\hat \Pi\Phi(H)&:=\pi\Phi(H)+\sum_{j\geq 1}(1-\delta(\epsilon_j\lvert x_a\rvert))\hat\pi_j(x_a,D_a)\Phi(H),\\
\Phi(H)\hat \Pi&:=\Phi(H)\pi+\sum_{j\geq 1}(1-\delta(\epsilon_j\lvert x_a\rvert))\Phi(H)\hat\pi_j(x_a,D_a),
\end{align*}
and
$$\hat \Pi_{\Phi}:=\Phi(H)\hat \Pi+(1-\phi(H))\hat \Pi\Phi(H)+(1-\phi(H))\pi(1-\phi(H)).$$
Then we have \eqref{myeq2.1} with $\Pi$ replaced by $\hat \Pi_{\Phi}$ and $\lVert\hat\Pi_{\Phi}-\pi\rVert_{\mathcal L(\mathcal H)}<1/2$. We define
$$\Pi:=\oint_{\lvert z-1\rvert=1/2}(\hat\Pi_{\Phi}-z)^{-1}dz.$$
Then as in Ashida \cite{As} we have \eqref{myeq2.0}, \eqref{myeq2.1} and \eqref{myeq2.1.1}.
\end{proof}

The proofs of Theorem \ref{ais1} and Theorem \ref{ais2} are similar to the ones of Theorem \ref{main1}, Theorem \ref{main2} and Theorem \ref{main3}. The difference is that we have $\lvert E_2\rvert=O(\tau^{-\infty})$ when $\pi_{\alpha}$ is replaced by $\Pi$, so that we can repeat the inductive argument without the restriction $\gamma_0<\gamma'_0<2\rho$.

\begin{proof}[Idea of the proof of Theorem \ref{three}]
Using \cite[Theorem 2.3]{Wa2} we can insert $\pi$ in front of $e^{-itH}$ as in the proof of \cite[Theorem 2.5]{Wa2}. Since $\langle x\rangle^{1+\rho}(\pi-\Pi)\in \mathcal L(\mathcal H)$, using Theorem \ref{ais1} we only need to prove \eqref{myeq2.1.2} for $p\in S_{-1-\rho}^0$. Repeating the same procedure we can assume $p\in S_{-n(1+\rho)}^0$ for any $n\in \mathbb N$. Choosing $n$ such that $n(1+\rho)\geq s$, by Lemma \ref{Skibstedlem} we obtain \eqref{myeq2.1.2}.
\end{proof}

\section{Concluding remarks}\label{concluding}
Here we discuss about interpretations on the conditions for the indices in propagation estimates.
In general the condition $s<\rho$ does not seem to be removable. The reason why the restriction $s<\rho$ occur seems to be different from that for the restriction $s<s'$.

(i) The restriction $s<s'$ in our Theorem 2.2 or Theorem 3.3 in Skibsted \cite{Sk} seems to be needed because if $x\in X$ is large in an initial state at time $t=0$, $x$ may decrease as time passes. In that case, the state starts scattering only after $x$ become small. If $\lvert x\rvert =r$ initially and the speed of the particles is $v$ which is determined from the kinetic energy, the time needed for particles to start scattering may be $t=rv^{-1}$. Therefore, if the distribution of the initial state is as $\langle x\rangle^{-s'}$, the amount (norm) of the part of the wave function which have not started to scatter at time $t$ may be estimated as $\langle vt\rangle^{-s'}$, so that this part does not affect the propagation estimates.

(ii) The bound $t^{-\rho}$ seems to come from the interaction between different scattering channels. The projected wave function $\pi_{\alpha}\psi(t)$ where $$\psi(t)=e^{-itH}f(H)\langle x\rangle^{-s'}\psi$$
for $\psi\in \mathcal H$ can be written as $\phi_{\alpha}(x^a)\otimes \psi_{\alpha}(x_a,t)$ where $\phi_{\alpha}(x^a)$ is the eigenfunction corresponding to the channel $\alpha$ and $\psi_{\alpha}(x_a,t)\in L^2(X_a)$. If the part $\phi_{\alpha}(x^a)\otimes \psi_{\alpha}(x_a,t)$ did not interact with the other parts corresponding to the other scattering channels such as $\phi_{\beta}(x^b)\otimes \psi_{\beta}(x_b,t)$, $x_a$ would behave as $\lvert x_a \rvert\sim 2(E-\lambda_{\alpha})^{1/2}t$ since the influence of $\tilde V_a(x_a)$ decays as $\psi_{\alpha}(x_a,t)$ scatters. However in reality, the part $\phi_{\alpha}(x^a)\otimes \psi_{\alpha}(x_a,t)$ interacts with the other scattering channels through $\tilde V_a(x_a)$ which is the only term in $H$ that does not commute with $\pi_{\alpha}$ (note that if $[\pi_{\alpha},H]=0$ then $\mathrm{Ran} \, \pi_{\alpha}$ is conserved).

A state in channels $\beta$ with $\lambda_{\beta}>\lambda_{\alpha}$ scatters at the speed $2(E-\lambda_{\beta})^{1/2}$ which is smaller than $2(E-\lambda_{\alpha})^{1/2}$. Therefore if this state is converted into a state in $\mathrm{Ran}\, \pi_{\alpha}$, it can not be cut off by $\chi\left(\frac{x_a^2}{4t^2}-(E-\lambda_{\alpha})<-\epsilon\right)$ for $\epsilon$ small enough.

We obtain a rough estimate on the rate of the conversion as follows. When the state $\phi_{\beta}(x^b)\otimes \psi_{\beta}(x_b,t)$ is converted into $\mathrm{Ran}\, \pi_{\alpha}$, the particles in the same cluster $C_0$ of $a$ may be close to each other. Moreover, since the kinetic energy of every channel is larger than or equal to $E$, we have $\lvert x\rvert >2Et$. Thus $x_{ij}$ for $(ij)\nleq a$ necessarily satisfies $C^{-1}t<\lvert x_{ij}\rvert<Ct$ for some constant $C>0$, since $\lvert x_{ij}\rvert\geq C_1\lvert x\rvert-C_2\lvert x^a\rvert$ for some $C_1,C_2>0$. Therefore we have $\lvert\tilde V_a(x_a)\rvert\sim t^{-\rho}$. Thus the part of the state $\phi_{\beta}(x^b)\otimes \psi_{\beta}(x_b,t)$ converted into $\mathrm{Ran} \, \pi_{\alpha}$ during the time interval $[t,t+\Delta t]$ is estimated as $\sim \tilde \lvert V_a\rvert\Delta t\sim t^{-\rho}\Delta t$ which implies the bound $t^{-\rho}$ in our propagation estimates. 

In more detail we can explain the estimate above as follows. When $\phi_{\beta}(x^b)\otimes \psi_{\beta}(x_b,t)$ is regarded as an element in $L^2(X^a)\otimes L^2(X_a)=L^2(X)$ it may include the state $\phi(x^a)\otimes\psi_1(x_a,t)$ with $\phi(x^a)\in \mathcal H_{ac}(H^a)$  or $\phi(x^a)\in\mathrm{Ran}\, \pi_{\alpha'}$ and $\psi_1(x_a,t)\in L^2(X_a)$ where $\mathcal H_{ac}(H^a)$ is the absolutely continuous subspace of $H^a$ and $\alpha'\neq \alpha$ is a channel. If $\tilde V_a=0$, $\phi(x^a)\otimes\psi_1(x_a,t)$ cannot be converted into $\phi_{\alpha}(x^a)\otimes \psi_2(x_a,t)$ with $\psi_2(x_a,t)\in L^2(X_a)$, since $\phi(x^a)$ and $\phi_{\alpha}(x^a)$ are in the subspaces corresponding to the different spectrum of $H^a$. The transition between different states $u$ and $v$ of unperturbed Hamiltonian ($H^a$ in the case above) by a perturbation $V$ ($\tilde V_a$ in the case above) during $[t,t+\Delta t]$ is estimated by $(v,Vu)\Delta t$ (see, e.g., Schiff \cite{Sc}).

However, for the process as above to happen effectively, it seems that $H^a$ need to have negative thresholds besides the negative eigenvalues of $H^a$, that is, for some $b\lneqq a$, $H^b$ need to have negative eigenvalues, because otherwise the estimates for $s>\rho$ hold (cf. G\'erard \cite[Theorem 2.2 (ii)]{Ge} and Wang \cite[Theorem 2.5]{Wa2}).

\subsection*{Acknowledgment}
This work was supported by JSPS KAKENHI Grant Number JP16J05967.

\appendix
\renewcommand{\theequation}{A.\arabic{equation}}
\setcounter{equation}{0}
\renewcommand{\thethm}{A.\arabic{thm}}
\setcounter{thm}{0}
\section{Lemmas concerned with domains of operators}
We assume throughout the appendix that $\mathcal H=L^2(X)$, $H$ is a N-body Hamiltonian and $a$ is a cluster decomposition with $\#a=2$.
\begin{as}\label{commutatoras}
Let $A$ be a selfadjoint operator on $\mathcal H$ such that $\forall n\in \mathbb N,\ \mathcal S=\mathcal S(X)\subset\mathcal D(A)\cap\mathcal D(H)\cap \mathcal D(\lvert x^a\rvert^n)$. Assume moreover\\
(1) With $\mathrm{ad}^0_A(H)=H$ and for any $n\in \mathbb N$ the form (defined iteratively)  $i^n\mathrm{ad}_A^n(H)=i[i^{n-1}\mathrm{ad}^{n-1}_A(H),A]$ on $\mathcal S$ extends to a symmetric operator with domain $\mathcal D(H)\cap \mathcal D(\lvert x^a\rvert^n)$.\\
(2) For any $t\in \mathbb R$ such that $\lvert t\rvert <1$ and any $n\in \mathbb N,\ e^{itA}(\mathcal D(H)\cap \mathcal D(\lvert x^a\rvert^n))\subset \mathcal D(H)\cap \mathcal D(\lvert x^a\rvert^n),\ e^{itA}\mathcal S\subset \mathcal S$ and there exists $C>0$ such that for any $u\in\mathcal D(H)\cap \mathcal D(\lvert x^a\rvert^n),\ \lVert \mathrm{ad}^n_A(H)e^{itA}u\rVert<C$.
\end{as}

\begin{lem}\label{commutatorlem}
Suppose Assumption \ref{commutatoras}. Then for any $u,v \in \mathcal D(H)\cap\mathcal D(A)\cap\mathcal D(\lvert x^a\rvert^{n+1})$
$$i(i^n\mathrm{ad}^n_A(H)u,Av)-i(Au,i^n\mathrm{ad}^n_A(H)v)=(u,i^{n+1}\mathrm{ad}^{n+1}_A(H)v)$$
\end{lem}
\begin{proof}
We define $(ad_A^n(H))(t):=e^{-itA}ad_A^n(H)e^{itA}$ which has the domain including $\mathcal D(H)\cap\mathcal D(\lvert x^a\rvert^n)$ by the Assumption \ref{commutatoras}(2). Then, since
$$i^n\mathrm{ad}_A^n(H)e^{i(t+s)A}u,$$
is bounded uniformly with respect to $t$ and weekly continuous, we have for any $s,t\in \mathbb R$ and $u,v \in \mathcal D(H)\cap\mathcal D(A)\cap\mathcal D(\lvert x^a\rvert^{n+1})$,
\begin{equation}\label{myeq ap0}
\begin{split}
\lim_{t\to 0}\frac{1}{t}((i^n(\mathrm{ad}_A^n(H))&(t+s)-i^n(\mathrm{ad}_A^n(H))(s))u,v)\\
&=\lim_{t\to 0}\frac{1}{t}((i^n\mathrm{ad}_A^n(H)e^{i(t+s)A}u,(e^{i(t+s)A}-e^{isA})v)\\
&\quad +(i^n\mathrm{ad}_A^n(H)(e^{i(t+s)A}-e^{isA})u,e^{isA}v)).\\
&=i(i^n\mathrm{ad}^{n+1}_A(H)e^{isA}u,Ae^{isA}v)\\
&\quad -i(Ae^{isA}u,i^n\mathrm{ad}^n_A(H)e^{isA}v).
\end{split}
\end{equation}

Therefore for any $u,v\in\mathcal S$ we have
\begin{equation*}
\frac{d}{dt}(i^n\mathrm{ad}_A^n(H)(t)u,v)=(e^{itA}u,i^{n+1}\mathrm{ad}^{n+1}_A(H)e^{itA}v).
\end{equation*}
Integrating from $0$ to $t$ we have
\begin{equation}\label{myeq ap1}
\frac{1}{t}((i^n\mathrm{ad}_A^n(H)(t)-i^n\mathrm{ad}_A^n(H))u,v)=\frac{1}{t}\int_0^t(e^{isA}u,i^{n+1}\mathrm{ad}^{n+1}_A(H)e^{isA}v)ds.
\end{equation}
For $u,v \in\mathcal D(H)\cap\mathcal D(A)\cap\mathcal D(\lvert x^a\rvert^{n+1})$ we have a sequence $u_m,v_k\in \mathcal S$ such that $u_m\to u,\ \mathrm {as}\ m\to \infty,\ v_k\to v,\ \mathrm {as}\ k\to \infty$. Letting $m\to\infty,\ k\to \infty$ in order in the equation \eqref{myeq ap1} for $u_m,v_k$, we obtain \eqref{myeq ap1} for $u,v \in\mathcal D(H)\cap\mathcal D(A)\cap\mathcal D(\lvert x^a\rvert^{n+1})$. Letting $t\to 0$, by \eqref{myeq ap0} the statement follows.
\end{proof}

\begin{as}\label{commutatoras2}
Let $t_0>0$ and $A(\tau)$ be a selfadjoint operator satisfying Assumption \ref{commutatoras} for all $\tau>t_0$. Assume moreover\\
(1)If $m\geq n$,\ then for any $u\in\mathcal D(H),\ \mathrm{ad}_{A(\tau)}^n(H)\langle x^a\rangle^{-m}\mathcal D(H)\subset \mathcal D(\lvert x^a\rvert^{m-n})$ and $\langle x^a\rangle^{m-n}\mathrm{ad}_{A(\tau)}^n(H)\langle x^a\rangle^{-m}(H+i)^{-1}$ is a continuous $\mathcal L(\mathcal H)$-valued function of $\tau$. Moreover,
$$\langle x^a\rangle^{m-n}\mathrm{ad}_{A(\tau)}^n(H)\langle x^a\rangle^{-m}(H+i)^{-1}=O(1),\ \mathrm {as}\ \tau\to \infty.$$
(2)For any $n\geq 0,\ (A(\tau)+i\lambda)^{-1}\mathcal D(\langle x^a\rangle^n)\subset\mathcal D(\langle x^a\rangle^n)$ and $(A(\tau)+i\lambda)^{-1}\mathcal D(H)\subset\mathcal D(H)$ for sufficiently large real $\lambda$ and for any $u\in\mathcal D(\langle x^a\rangle^n)$,
$$\mathrm{ad}_{A(\tau)}^n(H)i\lambda(A(\tau)+i\lambda)^{-1}(H+i)^{-1}u\to\mathrm{ad}_{A(\tau)}^n(H)(H+i)^{-1}  u,\ \mathrm{as}\ \lvert\lambda\rvert \to \infty.$$
\end{as}

\begin{lem}\label{commutatorlem2}
Suppose Assumption \ref{commutatoras2}.
Then for all $z$ such that $\mathrm{Im}z\neq 0$, the form
$i^n\mathrm{ad}^n_{A(\tau)}((H+z)^{-1})$ (resp., $i^n\mathrm{ad}^n_{A(\tau)}(f(H))$) which is defined iteratively
with $f\in C_0^{\infty}(\mathbb R)$) on $D({A(\tau)})\cap D(\lvert x^a\rvert^n)$ extends to a symmetric operator on $D(\lvert x^a\rvert^n)$ such that for all $u\in D(\lvert x^a\rvert^n)$, $\mathrm{ad}_{A(\tau)}^n((H+z)^{-1})u$ ($\mathrm{resp}., \mathrm{ad}_{A(\tau)}^n(f(H))u$) is continuous with respect to $\tau$ and
$\mathrm{ad}_{A(\tau)}^n((H+z)^{-1})u=O(1)$ (resp., $\mathrm{ad}_{A(\tau)}^n(f(H))u=O(1)$), $\ \mathrm{as}\ \tau\to \infty$.
\end{lem}

\begin{proof}
We prove by the induction in $n$ that $\mathrm{ad}^n_{A(\tau)}((H+z)^{-1})$ is the sum of terms of the form
\begin{equation}\label{myeqap.1}
C(H+z)^{-1}\mathrm{ad}_{A(\tau)}^{n_1}(H)(H+z)^{-1}\dotsm\mathrm{ad}_{A(\tau)}^{n_k}(H)(H+z)^{-1},
\end{equation}
where $\sum_{j=1}^kn_j=n$ and $C$ is a constant. Note that by Assumption \ref{commutatoras2} and that for any $n\geq 0$ we have $(H+i)^{-1}D(\lvert x^a\rvert^n)\subset \mathcal D(\lvert x^a\rvert^n)$, the domain of \eqref{myeqap.1} includes $\mathcal D(\lvert x^a\rvert^n),\ n=\sum_{j=1}^kn_j$, and moreover,
$$\langle x^a\rangle^{m}(H+z)^{-1}\mathrm{ad}_{A(\tau)}^{n_1}(H)\dotsm(H+z)^{-1}\mathrm{ad}_{A(\tau)}^{n_k}(H)(H+z)^{-1}u=O(1),$$
$\mathrm{as}\ \tau\to \infty$ for any $m\leq n$.

For $n=0$ this statement holds. Suppose the statement for $n$. We set $B_j:=\mathrm{ad}_{A(\tau)}^{n_j}(H)(H+z)^{-1}$. Then as a form on $\mathcal D(A(\tau))\cap\mathcal D(\lvert x^a\rvert^n)$,
\begin{align*}
[(H+z)^{-1}\prod_{j=1}^mB_j,A(\tau)]&=\lim_{\lambda\to \infty}\left[(H+z)^{-1}\prod_{j=1}^mB_j,\frac{i\lambda A(\tau)}{A(\tau)+i\lambda}\right]\\
&=\lim_{\lambda\to \infty}\left[(H+z)^{-1}\prod_{j=1}^mB_j,\frac{\lambda^2}{A(\tau)+i\lambda}\right]\\
&=\lim_{\lambda\to \infty}\left[(H+z)^{-1},\frac{\lambda^2}{A(\tau)+i\lambda}\right]\prod_{j=1}^mB_j\\
&\quad+\sum_{j=1}^m\lim_{\lambda\to \infty}(H+z)^{-1}B_1\dotsm B_{j-1}\\
&\qquad\cdot\left[B_j,\frac{\lambda^2}{A(\tau)+i\lambda}\right]B_{j+1}\dotsm B_m.
\end{align*}
Here we note that by Assumption \ref{commutatoras2}
$$B_{j+1}\dotsm B_m\mathcal D(\lvert x^a\rvert^{n+1})\subset\mathcal D(\lvert x^a\rvert^{n+1-\sum_{\ell=j+1}^{j+1}n_{\ell}})\subset\mathcal D(\lvert x^a\rvert^{n_j+1}).$$
$$(B_1\dotsm B_{j-1})^*(H-\bar z)^{-1}\mathcal D(\lvert x^a\rvert^{n+1})\subset\mathcal D(\lvert x^a\rvert^{n_j+1}).$$
Thus by Lemma \ref{commutatorlem} we have
\begin{align*}
&\lim_{\lambda\to \infty}(H+z)^{-1}B_1\dotsm B_{j-1}\left[B_j,\frac{\lambda^2}{A(\tau)+i\lambda}\right]B_{j+1}\dotsm B_m\\
&\quad=\lim_{\lambda\to \infty}\Big(-(H-z)^{-1}B_1\dotsm B_{j-1}\frac{\lambda}{A(\tau)+i\lambda}\mathrm{ad}_{A(\tau)}^{n_j+1}(H)\frac{\lambda}{A(\tau)+i\lambda}(H+z)^{-1}\\
&\qquad\cdot B_{j+1}\dotsm B_m+(H-z)^{-1}B_1\dotsm B_{j-1}\mathrm{ad}_{A(\tau)}^{n_j}(H)(H+z)^{-1}\frac{\lambda}{A(\tau)+i\lambda}\\
&\qquad\quad\cdot\mathrm{ad}_{A(\tau)}^1(H)\frac{\lambda}{A(\tau)+i\lambda}(H+z)^{-1}B_{j+1}\dotsm B_m\Big)\\
&\quad=(H-z)^{-1}B_1\dotsm B_{j-1}\mathrm{ad}_{A(\tau)}^{n_j+1}(H)(H+z)^{-1}B_{j+1}\dotsm B_m\\
&\qquad-(H-z)^{-1}B_1\dotsm B_{j-1}\mathrm{ad}_{A(\tau)}^{n_j}(H)(H+z)^{-1}\mathrm{ad}_{A(\tau)}^1(H)(H+z)^{-1}B_{j+1}\dotsm B_m.
\end{align*}
Thus the statement holds for $n+1$.

By Assumption \ref{commutatoras2} and the form \eqref{myeqap.1}, for all $u\in D(\lvert x^a\rvert^n)$, $\mathrm{ad}_{A(\tau)}^n((H+z)^{-1})u$ is continuous with respect to $\tau$ and $\mathrm{ad}_{A(\tau)}^n((H+z)^{-1})u=O(1),\ \mathrm{as}\ \tau\to \infty$.

The result for $\mathrm{ad}_{A(\tau)}^n(f(H))$ follows from that of $\mathrm{ad}_{A(\tau)}^n((H+z)^{-1})$ and the formula
\begin{equation}\label{HS formula}
f(H)=\frac{1}{2\pi i}\int_{\mathbb C}\bar \partial_zF(z)(z-H)^{-1}dz\wedge d\bar z
\end{equation}
where $F$ is an almost analytic extension of $f$ (see Helffer-Sj\"ostrand \cite{HeSj}).
\end{proof}

\begin{lem}\label{commHeSj}
Let $P\in \mathcal L(\mathcal H)$ and $A$ be a selfadjoint operator on $\mathcal H$. Assume that for any $n\in \mathbb N$ the form (defined iteratively) $i^n\mathrm{ad}^n_{A}(P)$ on $D(A)\cap D(\lvert x^a\rvert^n)$ extends to a symmetric operator on $D(\lvert x^a\rvert^n)$ and for any $z,\ \mathrm{Im}\ z\neq0$, $[\langle x^a\rangle^{-1},(A-z)^{-1}]=0$. Then we have for any $g\in C_0^{\infty}(\mathbb R)$, $k\in \mathbb N$,
\begin{align*}
i[P,g(A)]=&i\sum_{m=1}^{k}\frac{(-1)^{m+1}}{m!}\mathrm{ad}^m_A(P)g^{(m)}(A)\\
&+\frac{1}{2\pi}\int_{\mathbb C}\bar\partial_zG(z)(A-z)^{-1}\mathrm{ad}^{k+1}_A(P)(A-z)^{-(k+1)}dz\wedge d\bar z,
\end{align*}
on $\mathcal D(\langle x^a\rangle^{k+1})$ where $G$ is an almost analytic extension of $g$.
\end{lem}
\begin{proof}
As a form on $\mathcal D(\langle x^a\rangle)$
$$[P,(A-z)^{-1}]=-(A-z)^{-1}\mathrm{ad}_A^1(P)(A-z)^{-1}.$$
This identity extends to an identity between operators on $\mathcal D(\langle x^a\rangle)$. Iterating the commutation we have
$$[P,(A-z)^{-1}]=-\sum_{m=1}^{k}\mathrm{ad}^m_A(P)(A-z)^{-(m+1)}-(A-z)^{-1}\mathrm{ad}^{k+1}_A(P)(A-z)^{-(k+1)},$$
on $\mathcal D(\langle x^a\rangle^{k+1})$. Thus by \eqref{HS formula} we obtain the result.
\end{proof}

\begin{lem}\label{communi}
Suppose $A$ satisfies Assumption \ref{commutatoras} and for all $n\in \mathbb N$, $\mathcal S$ is a core of $A^n$. Assume moreover, for any $m\geq n$ and $u\in\mathcal D(H),\ \mathrm{ad}_{A}^n(H)\langle x^a\rangle^{-m}\mathcal D(H)\subset \mathcal D(\lvert x^a\rvert^{m-n})$, and for any $r \geq 0,\ n\in\mathbb N,\ \langle x^a\rangle^rA^n(H+i)^{-n}\langle x\rangle^{-n-r}\in \mathcal L(\mathcal H)$. Then, for any $f\in C_0^{\infty}(\mathbb R)$, $e^{-itH}f(H)\langle x\rangle^{-s}\mathcal H\subset \mathcal D(\langle A\rangle^s)$ and $\langle A\rangle^se^{-itH}f(H)\langle x\rangle^{-s}$ is a continuous $\mathcal L(\mathcal H)$-valued function of $t$.
\end{lem}
\begin{proof}
By an interpolation argument, it suffices to show for any $n\in \mathbb N$,
$$e^{-itH}f(H)\langle x\rangle^{-n}\mathcal H\subset D(\langle A\rangle^n),$$
and $\langle A\rangle^ne^{-itH}f(H)\langle x\rangle^{-n}$ is a continuous $\mathcal L(\mathcal H)$-valued function of $t$.
We have
$$A^n e^{-itH}f(H)\langle x\rangle^{-n}u=A^ne^{-itH}f(H)(H+i)^{n}(H+i)^{-n}\langle x\rangle^{-n}u,\ \forall u\in \mathcal H.$$
Since Lemma \ref{commutatorlem2} can be applied to $A$ as a constant function of $\tau$, by the continuity of the almost analytic extension of $e^{-itx}f(x)(x+i)^{n}$ with respect to $t$, we have as a form on $\mathcal S$
\begin{align*}
A^ne^{-itH}f(H)(H+i)^{n}&=\sum_{m=0}^nc_mad_A^m(e^{-itH}f(H)(H+i)^{n})A^{n-m}\\
&=\sum_{m=0}^nC_m(t)A^{n-m},
\end{align*}
where $c_m$ is a constant and $C_m(t)$ is an operator such that $C_m(t)\langle x^a\rangle^{-m}$ and $C_m(t)^*\langle x^a\rangle^{-m}$ are bounded and continuous $\mathcal L(\mathcal H)$-valued functions of $t$.

For any $v\in \mathcal D(A^n)$ there exists a sequence $v_l\in \mathcal S$ such that $(A^n+i)(v-v_l)\to 0$ and since by the assumption for any $w\in \mathcal H,\ u=(H+i)^{-n}\langle x\rangle^{-n}w\in \mathcal D(A^n)$, there exists a sequence $u_k\in \mathcal S$ such that $(A^n+i)(u-u_k)\to 0$. We have for this $v_l$ and $u_k$,
$$(A^nv_l,e^{-itH}f(H)(H+i)^{n}u_k)=(v_l,\sum_{m=0}^nC_m(t)A^{n-m}u_k).$$

Since by the assumption we have $\langle x^a\rangle^{m}A^{n-m}(H+i)^{-n}\langle x\rangle^{-n}\in \mathcal L(\mathcal H)$, letting $k\to \infty$, $l\to \infty$ in order we have 
$$(A^nv,e^{-itH}f(H)(H+i)^{n}u)=(v,\sum_{m=0}^nC_m(t)A^{n-m}u).$$

Since $A^n$ is self-adjoint, this means that $e^{-itH}f(H)(H+i)^{n}u\in \mathcal D(A^n)$ and
$$A^ne^{-itH}f(H)(H+i)^{n}u=\sum_{m=0}^nC_m(t)A^{n-m}u,$$
which is a continuous $\mathcal L(\mathcal H)$-valued function of $t$.
\end{proof}

\end{document}